\documentclass{llncs}

\newcommand{\dia}{\hfill{$\diamond$}}

\usepackage{amssymb,amsmath,graphicx}
\usepackage{pgf}
\usepackage{color,enumerate,comment,boxedminipage}
\usepackage{float}
\usepackage{hyperref}
\usepackage{comment}
\usepackage{tikz} 
\usepackage[capitalise,noabbrev]{cleveref} 
\usetikzlibrary{shapes,snakes}
\usepackage[T1]{fontenc}

\newcommand{\problemdef}[3]{
	\begin{center}
		\begin{boxedminipage}{.99\textwidth}
			\textsc{{#1}}\\[2pt]
			\begin{tabular}{ r p{0.8\textwidth}}
				\textit{~~~~Instance:} & {#2}\\
				\textit{Question:} & {#3}
			\end{tabular}
		\end{boxedminipage}
	\end{center}
}

\newtheorem{open}[example]{Open Problem}

\newcommand{\ssi}{\subseteq_i}
\newcommand{\si}{\supseteq_i}

\newcommand{\SFVS}{\textsc{Subset Feedback Vertex Set}}
\newcommand{\SOCT}{\textsc{Subset Odd Cycle Transversal}}

\newcommand{\NP}{{\sf NP}}
\newcommand{\XP}{{\sf XP}}

\newcounter{daggerfootnote}
\newcommand*{\daggerfootnote}[1]{
    \setcounter{daggerfootnote}{\value{footnote}}
    \renewcommand*{\thefootnote}{\ensuremath{\dagger}}
    \footnote[2]{#1}
    \setcounter{footnote}{\value{daggerfootnote}}
    \renewcommand*{\thefootnote}{\arabic{footnote}}
    }

\title{Computing Subset Transversals in $H$-Free Graphs\protect\daggerfootnote{This paper received support from the Leverhulme Trust (RPG-2016-258). An extended abstract of it appeared in the proceedings of WG 2020~\cite{BJPP20}.}}

\author{Nick Brettell\inst{1}
\and
Matthew Johnson\inst{2}
\and
Giacomo Paesani\inst{2}
\and
Dani\"el Paulusma\inst{2}}

\institute{School of Mathematics and Statistics, Victoria University of Wellington, New Zealand\\
\texttt{nick.brettell@vuw.ac.nz}
\and
Department of Computer Science, Durham University, UK\\
\texttt{\{matthew.johnson2,giacomo.paesani,daniel.paulusma\}@durham.ac.uk}
}

\begin{document}

\maketitle

\begin{abstract}
We study the computational complexity of two well-known graph transversal problems, namely {\sc Subset Feedback Vertex Set} and {\sc Subset Odd Cycle Transversal}, by restricting the input to $H$-free graphs, that is, to graphs that do not contain some fixed graph~$H$ as an induced subgraph. By combining known and new results, we determine the computational complexity of both problems on $H$-free graphs for every graph $H$ except when $H=sP_1+P_4$ for some $s\geq 1$. As part of our approach, we introduce the {\sc Subset Vertex Cover} problem and prove that it is polynomial-time solvable for $(sP_1+P_4)$-free graphs for every $s\geq 1$.

\end{abstract}

\section{Introduction}

The central question in Graph Modification is whether or not a graph~$G$ can be modified into a graph from a prescribed class~${\cal G}$ via at most $k$ graph operations from a prescribed set $S$ of permitted operations
such as vertex or edge deletion.
The \emph{transversal} problems {\sc Vertex Cover}, {\sc Feedback Vertex Set} and {\sc Odd Cycle Transversal} are classical problems 
of this kind.
For example, the {\sc Vertex Cover} problem
is equivalent to asking if one can delete at most~$k$ vertices to turn~$G$ into a member of the class of edgeless graphs. The problems {\sc Feedback Vertex Set} and {\sc Odd Cycle Transversal} ask if a graph~$G$
can be turned into, respectively,  a forest or a bipartite graph by deleting vertices.

We can relax the condition on belonging to a prescribed class to obtain some related \emph{subset transversal} problems.  We state these formally after some definitions.
For a graph~$G=(V,E)$ and a set $T \subseteq V$, a {\it $T$-cycle} is a (not necessarily induced) cycle of $G$ that intersects $T$.
A $T$-cycle is {\it odd} if it has an odd number of vertices.
A set $S_T\subseteq V$~is a {\it $T$-vertex cover} of $G$ if $S_T$ contains at least one of the two end-vertices for every edge incident to a vertex of $T$. A set $S_T\subseteq V$~is a {\it $T$-feedback vertex set} or an {\it odd $T$-cycle transversal} of $G$ if $S_T$ contains at least one vertex of every $T$-cycle, or every odd $T$-cycle, respectively. 
For example, let $G$ be a star with center vertex~$c$, whose leaves form the set $T$. Then, both $\{c\}=V\setminus T$ and~$T$ are $T$-vertex covers of $G$ but the first is considerably smaller than the second.  
See Figures~\ref{subset-house} and~\ref{f-example} for some more examples. 

\problemdef{{\sc Subset Vertex Cover}}{a graph $G=(V,E)$, a subset $T\subseteq V$ and a positive integer $k$.}{does $G$ have a $T$-vertex cover $S_T$ with $|S_T|\leq k$?} 

\problemdef{{\sc Subset Feedback Vertex Set}}{a graph $G=(V,E)$, a subset $T\subseteq V$ and a positive integer $k$.}{does $G$ have a $T$-feedback vertex set $S_T$ with $|S_T|\leq k$?} 

\problemdef{{\sc Subset Odd Cycle Transversal}}{a graph $G=(V,E)$, a subset $T\subseteq V$ and a positive integer $k$.}{does $G$ have an odd $T$-cycle transversal $S_T$ with $|S_T|\leq k$?} 

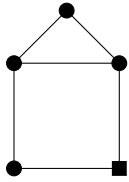
\begin{figure}
\begin{center}
\begin{tikzpicture}[scale=0.7]
\draw (-1,1)--(0,2)--(1,1)--(-1,1)--(-1,-1)--(1,-1)--(1,1);
\draw[fill=black] (-1,1) circle [radius=4pt] (0,2) circle [radius=4pt] (1,1) circle [radius=4pt] (-1,-1) circle [radius=4pt]  
(1,-1) node[regular polygon,regular polygon sides=4,draw,fill=black,scale=0.6pt] {};
\node[below] at (0,1) {};
\end{tikzpicture}
\caption{The house in which we let $T$ consist of the square vertex. The empty set is not an odd $T$-cycle transversal, as there exists a non-induced odd cycle (containing all the vertices of the graph) passing through the square vertex. Hence, the size of a minimum odd $T$-cycle transversal is $1$, and in particular, this example shows that looking for a minimum odd $T$-cycle transversal (the problem we consider) or a minimum odd {\it induced} $T$-cycle transversal are two different problems.} \label{subset-house}
\end{center} 
\end{figure}

\begin{figure}
\begin{center}
\begin{minipage}{0.45\textwidth}
\centering
\begin{tikzpicture}[xscale=0.65, yscale=0.65]
\draw (0,2)--(1.16,-1.6)--(-1.9,0.6)--(1.9,0.6)--(-1.16,-1.6)--(0,2) (-2.85,0.9)--(-1.74,-2.4)--(1.74,-2.4)--(2.85,0.9)--(0,3)--(-2.85,0.9) 
(-1.9,0.6)--(-2.85,0.9) (-1.16,-1.6)--(-1.74,-2.4) (1.16,-1.6)--(1.74,-2.4) (1.9,0.6)--(2.85,0.9) (0,2)--(0,3);
\draw[fill=white](-1.9,0.6) circle [radius=5pt] (-1.16,-1.6) circle [radius=5pt] (1.16,-1.6) circle [radius=5pt] (1.9,0.6) circle [radius=5pt] 
(0,2) node[regular polygon,regular polygon sides=4,draw,fill=black,scale=0.7pt] {} (-2.85,0.9) node[regular polygon,regular polygon sides=4,draw,fill=white,scale=0.7pt] {} 
(-1.74,-2.4) node[regular polygon,regular polygon sides=4,draw,fill=black,scale=0.7pt] {} (1.74,-2.4) node[regular polygon,regular polygon sides=4,draw,fill=white,scale=0.7pt] {} (0,3) circle [radius=5pt];
\draw[fill=black](2.85,0.9) circle [radius=5pt];
\end{tikzpicture}
\end{minipage}
\qquad
\begin{minipage}{0.45\textwidth}
\centering
\begin{tikzpicture}[xscale=0.65, yscale=0.65]
\draw (0,2)--(1.16,-1.6)--(-1.9,0.6)--(1.9,0.6)--(-1.16,-1.6)--(0,2) (-2.85,0.9)--(-1.74,-2.4)--(1.74,-2.4)--(2.85,0.9)--(0,3)--(-2.85,0.9)
(-1.9,0.6)--(-2.85,0.9) (-1.16,-1.6)--(-1.74,-2.4) (1.16,-1.6)--(1.74,-2.4) (1.9,0.6)--(2.85,0.9) (0,2)--(0,3);
\draw[fill=white] (-1.9,0.6) circle [radius=5pt] (-1.16,-1.6) circle [radius=5pt] (1.16,-1.6) circle [radius=5pt] (1.9,0.6) circle [radius=5pt]
(2.85,0.9) circle [radius=5pt] (0,2) node[regular polygon,regular polygon sides=4,draw,fill=black,scale=0.7pt] {}
(-2.85,0.9) node[regular polygon,regular polygon sides=4,draw,fill=black,scale=0.7pt] {} (-1.74,-2.4) node[regular polygon,regular polygon sides=4,draw,fill=white,scale=0.7pt] {}
(1.74,-2.4) node[regular polygon,regular polygon sides=4,draw,fill=black,scale=0.7pt] {}(0,3) circle [radius=5pt];
\end{tikzpicture}
\end{minipage}
\caption{In both examples, the square vertices of the Petersen graph form a set $T$ and the black vertices form an odd $T$-cycle transversal $S_T$, which is also a $T$-feedback vertex set. In the left example, $S_T\cap (V\setminus T)\neq \emptyset$, and in the right example, $S_T\subseteq T$.}\label{f-example}
\end{center}
\vspace*{-0.5cm}
\end{figure}
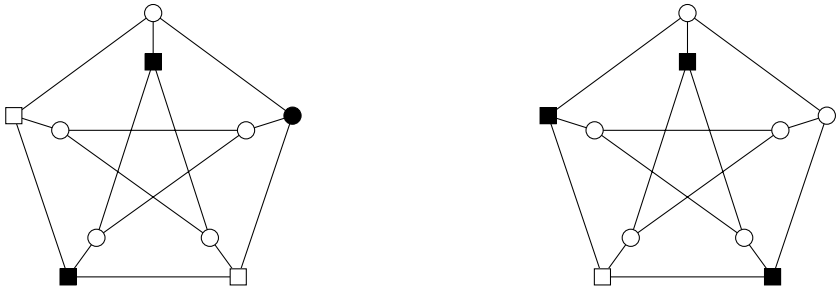

\noindent
The  
{\sc Subset Feedback Vertex Set} and {\sc Subset Odd Cycle Transversal} problems are well known. The {\sc Subset Vertex Cover} problem is introduced in this paper, and we are not aware of past work on this problem (see Section~\ref{s-con} for open problems). 
On general graphs, {\sc Subset Vertex Cover} is polynomially equivalent to {\sc Vertex Cover}: to solve {\sc Subset Vertex Cover} remove edges in the input graph that are not incident to any vertex of $T$ to yield an equivalent instance of {\sc Vertex~Cover}.
However, this equivalence no longer holds for graph classes that are {\it not} closed under edge deletion. 

As the three problems are \NP-complete, we consider the restriction of the input to special graph classes in order to better understand which graph properties cause the computational hardness. 
Instead of classes closed under edge deletion, we focus on classes of graphs closed under vertex deletion. Such classes are called {\it hereditary}.
The reasons for this choice are threefold. First, hereditary graph classes capture many well-studied graph classes. Second, every hereditary graph class~${\cal G}$ can be characterized by a (possibly infinite) set~${\cal F}_{\cal G}$ of forbidden induced subgraphs. This enables us to initiate a {\it systematic} study, starting from the case where
 $|{\cal F}_{\cal G}|=1$. Third, we aim to extend and strengthen existing complexity results (that are for hereditary graph classes).
If ${\cal F}_{\cal G}=\{H\}$ for some graph $H$, then 
${\cal G}$ is \emph{monogenic}, and 
 every $G\in {\cal G}$ is \emph{$H$-free}. Our research question is:

\medskip
\noindent
{\it How does the structure of a graph~$H$ influence the computational complexity of a subset transversal problem for input graphs that are $H$-free?}

\medskip
\noindent
As a general strategy one might first try to prove that the restriction to $H$-free graphs is \NP-complete if~$H$ contains a cycle or an induced claw (the 4-vertex star). 
This is usually done by showing, respectively, that the problem is \NP-complete on graphs of arbitrarily large girth (the length of a shortest cycle) and on line graphs, which form a subclass of claw-free graphs. If this is the case, then it remains to consider the case where~$H$ has no cycle, and has no claw either. So~$H$ is a \emph{linear forest}, that is, the disjoint union of one or more paths.

\smallskip
\noindent
{\bf Existing Results.} As \NP-completeness results for transversal problems carry over to subset transversal problems, we first discuss results on {\sc Feedback Vertex Set} and {\sc Odd Cycle Transversal} for $H$-free graphs.  By Poljak's construction~\cite{Po74}, {\sc Feedback Vertex Set}  is \NP-complete for graphs of girth at least~$g$ for every integer~$g\geq 3$. The same holds for {\sc Odd Cycle Transversal}~\cite{CHJMP18}. Moreover, {\sc Feedback Vertex Set}~\cite{Mu17b} and {\sc Odd Cycle Transversal}~\cite{CHJMP18} are \NP-complete for line graphs and thus for claw-free graphs. Hence, both problems are \NP-complete for $H$-free graphs if $H$ has a cycle or claw. Both problems are polynomial-time solvable for $P_4$-free graphs~\cite{BK85}, for $sP_2$-free graphs for every $s\geq 1$~\cite{CHJMP18} and for $(sP_1+P_3)$-free graphs for every $s\geq1$~\cite{DFJPPP19}. In addition, {\sc Odd Cycle Transversal} is \NP-complete for $(P_2+P_5,P_6)$-free graphs~\cite{DFJPPP19}.
Very recently, Abrishami et al. showed that {\sc Feedback Vertex Set} is polynomial-time solvable for $P_5$-free graphs~\cite{ACPRS20}.
We summarize as follows 
($F\ssi G$ means that $F$ is an induced subgraph of $G$;
see Section~\ref{s-pre} for the other notation used).

\begin{theorem}\label{t-known}
For a graph $H$, {\sc Feedback Vertex Set} on $H$-free graphs is 
polynomial-time solvable if 
$H\ssi P_5$,  $H\ssi sP_1+P_3$ or $H\ssi sP_2$ for some $s\geq 1$, and is \NP-complete if $H\si C_r$ for some $r\geq 3$ or $H\si K_{1,3}$.
\end{theorem}

\begin{theorem}\label{t-known2}
For a graph $H$, 
 {\sc Odd Cycle Transversal} on $H$-free~graphs is polynomial-time solvable if 
$H=P_4$, $H\ssi sP_1+P_3$ or $H\ssi sP_2$ for some~$s\geq 1$, and is \NP-complete  if $H\si C_r$ for some $r\geq 3$, $H\si K_{1,3}$, $H\si P_6$ or $H\si P_2+P_5$.
\end{theorem}

We note that no integer~$r$ is known such that {\sc Feedback Vertex Set} is \NP-complete for $P_r$-free graphs. This situation changes for {\sc Subset Feedback Vertex Set} which is, 
unlike {\sc Feedback Vertex Set}, 
\NP-complete for split graphs (that is, $(2P_2,C_4,C_5)$-free graphs), as shown by Fomin et al.~\cite{FHKPV14}.
Papadopoulos and Tzimas~\cite{PT19,PT20} proved that {\sc Subset Feedback Vertex Set} is polynomial-time solvable for $sP_1$-free graphs for any~$s\geq 1$, co-bipartite graphs, interval graphs and permutation graphs, and thus $P_4$-free graphs.
Some of these results were generalized by Bergougnoux et al.~\cite{BPT19}, who solved an open problem of Jaffke et al.~\cite{JKT20} by giving an $n^{O(w^2)}$-time algorithm for {\sc Subset Feedback Vertex Set} given a graph and a decomposition of this graph of mim-width~$w$. This does not lead to new results for $H$-free graphs: 
a class of $H$-free graphs has bounded mim-width if and only if $H\ssi P_4$~\cite{BHMPP}.

We are not aware of any results on {\sc Subset Odd Cycle Transversal} for $H$-free graphs, but note that 
this problem generalizes 
{\sc Odd Multiway Cut}, just as {\sc Subset Feedback Vertex Set} generalizes 
{\sc Node Multiway Cut}, another well-studied problem.
We refer to a large body of literature~\cite{BBBK20,CFLMRS17,CPPW13,FHKPV14,GHKS14,HK18,IWY16,KKK12,KK12,KW12,LMRS17} for further details, in particular for parameterized and exact algorithms for {\sc Subset Feedback Vertex Set} and {\sc Subset Odd Cycle Transversal}. These algorithms are beyond the scope of this paper.

\medskip
\noindent
{\bf Our Results.}
We significantly extend the known results for {\sc Subset Feedback Vertex Set} and {\sc Subset Odd Cycle Transversal} on $H$-free graphs. These new results lead us to the following two almost-complete dichotomies:

\begin{theorem}\label{t-main}
Let $H$ be a graph with $H\neq sP_1+P_4$ for all $s\geq 1$. 
Then {\sc Subset Feedback Vertex Set} on $H$-free graphs is polynomial-time solvable if 
$H=P_4$ or $H\ssi sP_1+P_3$ for some $s\geq 1$, and is \NP-complete otherwise.
\end{theorem}

\begin{theorem}\label{t-main2}
Let $H$ be a graph with $H\neq sP_1+P_4$ for all $s\geq 1$. 
Then {\sc Subset Odd Cycle Transversal} on $H$-free graphs is polynomial-time solvable if 
$H=P_4$ or $H\ssi sP_1+P_3$ for some $s\geq 1$ and  \NP-complete otherwise.
\end{theorem}

\begin{figure}
\centering
\begin{minipage}[b]{0.04\textwidth}
\begin{tikzpicture}[xscale=0.5, yscale=0.5] \draw (0,2.5)--(0,-0.5);
\draw[fill=black] (0,0.5) circle [radius=3pt] (0,1.5) circle [radius=3pt] (0,-0.5) circle [radius=3pt] (0,2.5) circle [radius=3pt];
\node at (0,-3) {$P_4$};
\end{tikzpicture}
\end{minipage}
\qquad
\begin{minipage}[b]{0.04\textwidth}
\begin{tikzpicture}[xscale=0.5, yscale=0.5] \draw (0,2)--(0,-2);
\draw[fill=black] (0,2) circle [radius=3pt] (0,1) circle [radius=3pt] (0,0) circle [radius=3pt] (0,-1) circle [radius=3pt] (0,-2) circle [radius=3pt];
\node at (0,-4) {$P_5$};
\end{tikzpicture}
\end{minipage}
\qquad
\begin{minipage}[b]{0.02\textwidth}
\begin{tikzpicture}[xscale=0.5, yscale=0.5] \draw (0,2.5)--(0,-2.5);
\draw[fill=black] (0,2.5) circle [radius=3pt] (0,1.5) circle [radius=3pt] (0,0.5) circle [radius=3pt] (0,-0.5) circle [radius=3pt] (0,-1.5) circle [radius=3pt] (0,-2.5) circle [radius=3pt];
\node at (0,-4) {$P_6$};
\end{tikzpicture}
\end{minipage}
\qquad
\begin{minipage}[b]{0.08\textwidth}
\begin{tikzpicture}[xscale=0.5, yscale=0.5] \draw (1,2.5)--(1,-1.5) (0,1.3)--(0,-0.3);
\draw[fill=black] (1,2.5) circle [radius=3pt] (1,1.5) circle [radius=3pt] (1,0.5) circle [radius=3pt] (1,-0.5) circle [radius=3pt] (1,-1.5) circle [radius=3pt] (0,1.3) circle [radius=3pt] (0,-0.3) circle [radius=3pt];
\node at (0,-3.5) {$P_2+P_5$};
\end{tikzpicture}
\end{minipage}
\qquad
\begin{minipage}[b]{0.08\textwidth}
\begin{tikzpicture}[xscale=0.5, yscale=0.5] \draw (0.7,2)--(0.7,-0.4) (-0.8,2.5)--(-1,2.5)--(-1,-0.9)--(-0.8,-0.9); \draw[fill=black] (0.7,2) circle [radius=3pt] 
(0.7,-0.4) circle [radius=3pt] (0.7,0.8) circle [radius=3pt] (-0.5,2.3) circle [radius=2.5pt] (-0.5,-0.7) circle [radius=2.5pt]  (-0.5,1.8) circle [radius=1.5pt] (-0.5,1.3) circle [radius=1.5pt] (-0.5,0.8) circle [radius=1.5pt] (-0.5,0.3) circle [radius=1.5pt] (-0.5,-0.2) circle [radius=1.5pt];
\node[left] at (-1,0.8) {$s$}; \node at (0,-3) {$sP_1+P_3$};
\end{tikzpicture}
\end{minipage}
\qquad
\begin{minipage}[b]{0.08\textwidth}
\begin{tikzpicture}[xscale=0.5, yscale=0.5] \draw (-0.5,1.5)--(0.5,1.5) 
(-0.5,-0.75)--(0.5,-0.75) (-0.5,2.25)--(0.5,2.25) (-0.8,2.45)--(-1,2.45)--(-1,-0.95)--(-0.8,-0.95);
\draw[fill=black] (-0.5,2.25) circle [radius=3pt] (-0.5,1.5) circle [radius=3pt] (-0.5,-0.75) circle [radius=3pt] 
(0.5,2.25) circle [radius=3pt] (0.5,1.5) circle [radius=3pt] (0.5,-0.75) circle [radius=3pt] 
(0,1.05) circle [radius=1.5pt] (0,0.6) circle [radius=1.5pt] (0,0.15) circle [radius=1.5pt] (0,-0.3) circle [radius=1.5pt];
\node[left] at (-1,0.65) {$s$};
\node at (0,-3) {$sP_2$};
\end{tikzpicture}
\end{minipage}
\qquad
\begin{minipage}[b]{0.13\textwidth}
\begin{tikzpicture}[xscale=0.4, yscale=0.4]
\draw (0,2)--(1.9,0.6)--(1.16,-1.6)--(-1.16,-1.6)--(-1.9,0.6)--(0,2); \draw[fill=black] (-1.9,0.6) circle [radius=3.6pt] 
(-1.16,-1.6) circle [radius=3.6pt] (1.16,-1.6) circle [radius=3.6pt] (1.9,0.6) circle [radius=3.6pt] (0,2) circle [radius=3.6pt];
\node at (0,-3.5) {$C_5$};
\end{tikzpicture}
\end{minipage}
\qquad
\begin{minipage}[b]{0.06\textwidth}
\begin{tikzpicture}[xscale=0.5, yscale=0.5]
\draw (-0.7,0)--(0.7,1.2) (-0.7,0)--(0.7,0) (-0.7,0)--(0.7,-1.2);
\draw[fill=black] (0.7,1.2) circle [radius=3pt] (0.7,-1.2) circle [radius=3pt] (0.7,0) circle [radius=3pt] (-0.7,0) circle [radius=3pt];
\node at (0,-3) {$K_{1,3}$};
\end{tikzpicture}
\end{minipage}
\caption{The forbidden graphs of Theorems~\ref{t-known}--\ref{t-main2}.}\label{notables}
\end{figure}
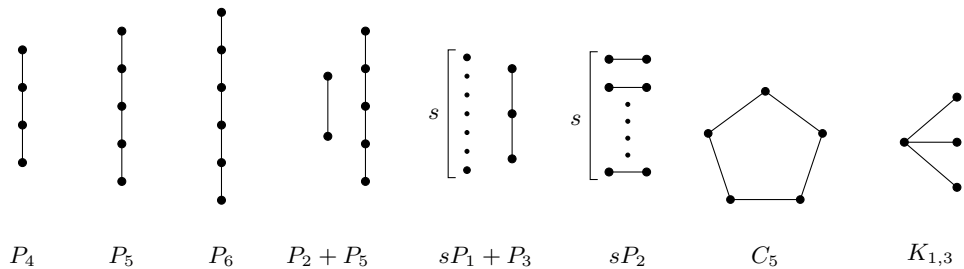

Though the proved complexities of {\sc Subset Feedback Vertex Set} and {\sc Subset Odd Cycle Transversal} are the same on $H$-free graphs, the algorithm that we present for \SOCT\ on $(sP_1+P_3)$-free graphs is more technical compared to the algorithm for \SFVS, and considerably generalizes the transversal algorithms for $(sP_1+P_3)$-free graphs of~\cite{DFJPPP19}.
There is further evidence that \SOCT\ is a more challenging problem than \SFVS.
For example, the best-known parameterized algorithm for {\sc Subset Feedback Vertex Set} runs in $O^*(4^k)$ time~\cite{IWY16}, but the best-known run-time for {\sc Subset Odd Cycle Transversal} is $O^*(2^{O(k^3 \log k)})$~\cite{LMRS17}.
Moreover, it is not known if there is an \XP\ algorithm for \SOCT\ in terms of mim-width in contrast to the known \XP\ algorithm for \SFVS~\cite{BPT19}.

In Section~\ref{s-pre} we introduce our terminology.
In Section~\ref{s-svc} we present some results for {\sc Subset Vertex Cover}: the first result shows that {\sc Subset Vertex Cover} is polynomial-time solvable for $(sP_1+P_4)$-free graphs for every $s\geq 1$, and we later use this as a subroutine to obtain a polynomial-time algorithm for \SOCT\ on $P_4$-free graphs.   
We present our results on {\sc Subset Feedback Vertex Set} and {\sc Subset Odd Cycle Transversal} in Sections~\ref{s-sfvs} and~\ref{s-soct}, respectively.
In~Section~\ref{s-con} on future work we discuss {\sc Subset Vertex Cover} in more detail.

\section{Preliminaries}
\label{s-pre}

We consider undirected, finite graphs with no self-loops and no multiple edges.
Let $G=(V,E)$ be a graph, and let $S\subseteq V$.
The graph~$G[S]$ is the subgraph~of~$G$ induced by~$S$. We write~$G-S$ to denote the graph $G[V\setminus S]$.
Recall that for a graph $F$, we write $F\ssi G$ if $F$ is an induced subgraph of $G$.
The cycle and path on $r$ vertices are denoted $C_r$ and $P_r$, respectively.
We say that~$S$ is {\it independent} if $G[S]$ is edgeless, and that $S$ is a {\it clique} if $G[S]$ is {\it complete}, that is, contains every possible edge between two vertices.
We let $K_r$ denote the complete graph on $r$ vertices, and $sP_1$ denote the graph whose vertices form an independent set of size~$s$. 
A {\it  (connected) component} of~$G$ is a maximal connected subgraph of $G$.
The graph $\overline{G}=(V,\{uv\; |\; uv\not \in E\; \mbox{and}\; u\neq v\})$ is the \emph{complement} of~$G$.
The \emph{neighbourhood} of a vertex $u\in V$ is the set $N_G(u)=\{v\; |\; uv\in E\}$. For $U\subseteq V$, we let $N_G(U)=\bigcup_{u\in U}N(u)\setminus U$. The {\it closed} neighbourhoods of $u$ and~$U$ are denoted
by $N_G[u]=N_G(u)\cup \{u\}$ and $N_G[U]=N_G(U)\cup U$, respectively.
We omit subscripts when there is no ambiguity.

Let $T\subseteq V$ be such that $S\cap T=\emptyset$.
Then~$S$ is \emph{complete} to~$T$ if every vertex of~$S$ is adjacent to every vertex of~$T$, and~$S$ is \emph{anti-complete} to~$T$ if there are no edges between~$S$ and~$T$.
In the first case, $S$ is also said to be \emph{complete} to~$G[T]$, and in the second case we say it is \emph{anti-complete} to~$G[T]$.

We say that $G$ is a \emph{forest} if it has no cycles, and, furthermore, that $G$ is a \emph{linear forest} if it is the disjoint union of one or more paths.
The graph~$G$ is \emph{bipartite} if $V$ can be partitioned into at most two independent sets.
A  graph is \emph{complete bipartite} if its vertex set can be partitioned into two independent sets~$X$ and~$Y$ such that
$X$ is complete to~$Y$. We denote such a graph by $K_{|X|,|Y|}$. If $X$ or $Y$ has size~$1$, the complete bipartite graph is a {\it star}; recall that $K_{1,3}$ is also called a claw.
A graph $G$ is a {\it split graph} if it has a bipartition $(V_1,V_2)$ such that $G[V_1]$ is a clique and $G[V_2]$ is an independent set.
A graph is split if and only if it is $(C_4, C_5, 2P_2)$-free~\cite{FH77}.

Let~$G_1$ and~$G_2$ be two vertex-disjoint graphs.
The \emph{union} operation $+$ creates the disjoint union $G_1+\nobreak G_2$ of~$G_1$ and~$G_2$ (recall that $G_1+G_2$ is the graph with vertex set $V(G_1)\cup V(G_2)$ and edge set $E(G_1)\cup E(G_2)$). 
The \emph{join} operation adds an edge between every vertex of~$G_1$ and every vertex of~$G_2$.
The graph~$G$ is a \emph{cograph} if~$G$ can be generated from~$K_1$ by a sequence of join and union operations.
A graph is a cograph if and only if it is $P_4$-free (see, for example,~\cite{BLS99}).
It is also well known~\cite{CLS81} that a graph $G$ is a cograph if and only if $G$ allows a unique tree decomposition called the {\it cotree} $T_G$ of $G$, which has the following properties:
\begin{itemize}
\item [1.]  The root $r$ of $T_G$ corresponds to the graph $G_r=G$.
\item [2.]   Each leaf $x$ of $T_G$ corresponds to exactly one vertex of $G$, and vice versa. Hence $x$ corresponds to a unique single-vertex graph $G_x$.
\item [3.]  Each internal node $x$ of $T_G$ has at least two children, is  labelled $\oplus$ or $\otimes$, and corresponds to an induced subgraph $G_x$ of $G$ defined as follows:
\begin{itemize}
\item if $x$ is a $\oplus$-node, then $G_x$ is the disjoint union of all graphs $G_y$ where $y$ is a child of~$x$;
\item if $x$ is a $\otimes$-node, then $G_x$ is the join of all graphs $G_y$ where $y$ is a child of~$x$.
\end{itemize}
\item [4.] Labels of internal nodes on the (unique) path from any leaf to $r$ alternate between $\oplus$ and~$\otimes$.
\end{itemize}
Note that $T_G$ has $O(n)$ vertices. We modify $T_G$ into a {\it modified cotree} $T_G'$ in which each internal node has exactly two children by applying the following well-known procedure (see for example~\cite{BM93}). If an internal node~$x$ of $T_G$ has more than two children $y_1$ and $y_2$, remove the edges $xy_1$ and $xy_2$ and add a new vertex $x'$ with 
edges $xx'$, $x'y_1$ and $x'y_2$. If~$x$ is a $\oplus$-node, then $x'$ is a $\oplus$-node. If $x$ is a $\otimes$-node, then $x'$ is a 
$\otimes$-node. Applying this rule exhaustively yields $T_G'$. As~$T_G$ has $O(n)$ vertices, constructing $T_G'$ from $T_G$ takes linear time. This leads to the following result, due to Corneil, Perl and Stewart, who proved it for cotrees.

\begin{lemma}[\cite{CPS85}]\label{l-cotree}
Let $G$ be a graph with $n$ vertices and $m$ edges.  Then deciding whether or not $G$ is a cograph, and constructing a modified cotree $T_G'$ (if it exists) takes time $O(n+m)$.
\end{lemma}

We  also consider optimization versions of
subset transversal problems, in which case
we have  instances~$(G,T)$  (instead of instances $(G,T,k)$).
We say that a set $S\subseteq V(G)$ is a {\it solution} for an instance $(G,T)$ if  $S$ is a $T$-transversal (of whichever kind we are concerned with).  A solution $S$ is {\it smaller} than a solution $S'$ if $|S|<|S'|$, and a solution $S$ is 
{\it minimum} if $(G,T)$ does not have a solution smaller than~$S$, and it is {\it maximum} if there is no larger solution.
We will use the following general lemma, which was implicitly used in~\cite{PT20}.

\begin{lemma}\label{bound}
Let $S$ be a minimum solution for an instance $(G,T)$ of a subset transversal problem.
Then $|S \setminus T| \le |T \setminus S|$.
\end{lemma}
\begin{proof}
For contradiction, assume that $|S \setminus T| > |T \setminus S|$. Then $|T|<|S|$
(see also \cref{scheme}).
This means that $T$ is a smaller solution than $S$, a contradiction. \qed
\end{proof}

\begin{figure}
\begin{center}
\begin{tikzpicture}[scale=0.4]
\draw (-10,5)--(10,5)--(10,-5)--(-10,-5)--(-10,5) (2,5)--(2,-5) (-10,-1)--(10,-1);
\node[above] at (-4.5,5) {$V\setminus S$};
\node[above] at (5.5,5) {$S$};
\node[left] at (-10,2) {$T$};
\node[left] at (-10,-3) {$V\setminus T$};
\node at (-4.5,2) {$T\setminus S$};
\node at (5.5,-3) {$S\setminus T$};
\end{tikzpicture}
\end{center}
\caption{For a minimum solution $S$ for an instance $(G,T)$ of a subset transversal problem, it must hold that $|S \setminus T| \le |T \setminus S|$  (see the proof of \cref{bound}).}
\label{scheme}
\end{figure}
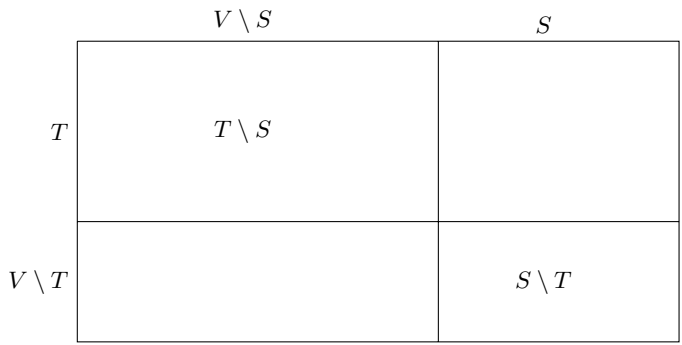

Let $T\subseteq V$ be a vertex subset of a graph $G=(V,E)$. 
Recall that a cycle is a $T$-cycle if it contains a vertex of $T$. 
A subgraph of $G$ is a {\it $T$-forest} if it has no $T$-cycles.
Recall also that a cycle is odd if it has an odd number of edges.
A subgraph of $G$ is {\it $T$-bipartite} if it has no odd $T$-cycles. 
We recall that a set $S_T\subseteq V$~is a {$T$-vertex cover} of $G$ if $S_T$ has at least one vertex of every edge incident to a vertex of $T$.
We also recall that a set $S_T\subseteq V$~is a {$T$-feedback vertex set} or an {odd $T$-cycle transversal} of $G$ if $S_T$ has at least one vertex of every $T$-cycle or every odd $T$-cycle, respectively.
Note that $S_T$ is a $T$-feedback vertex set if and only if $G[V\setminus S_T]$ is a $T$-forest, and 
 $S_T$ is an odd $T$-cycle transversal if and only if $G[V\setminus S_T]$ is $T$-bipartite.
A {\it $T$-path} is a path that contains a vertex of $T$.  A $T$-path is {\it odd} (or {\it even}) if the number of edges in the path is odd (or even, respectively).

We will use the following easy lemma, which proves that $T$-forests and $T$-bipartite graphs can be recognized in polynomial time.
It combines results claimed but not proved in~\cite{LMRS17,PT20}.

\begin{lemma}\label{st-test}
Let $G=(V,E)$ be a graph and $T\subseteq V$. Then deciding whether or not $G$ is a $T$-forest or $T$-bipartite takes $O(n+m)$ time.
\end{lemma}

\begin{proof}
A {\it block} of $G$ is a maximal 2-connected subgraph of $G$ and is {\it non-trivial} if it contains a cycle, or, equivalently, at least three vertices.
Suppose that we have a block decomposition of $G$; it is well known that this can be found in $O(n+m)$ time (see for example~\cite{HT73}).
It is clear that $G$ is a $T$-forest if and only if no non-trivial block contains a vertex of $T$.
We claim that $G$ is $T$-bipartite if and only if no non-bipartite block contains a vertex of $T$.  
To see this note first that the sufficiency is obvious.  We will show that if a vertex $t$ of $T$ belongs to a block $B$ that contains an odd cycle $C$, then $t$ belongs to an odd cycle.  If $t$ is in $C$, we are done.  Otherwise find two paths $P$ and $P'$ from $t$ to, respectively, distinct vertices $u$ and $u'$ in $C$.  We can assume that the paths contain no other vertex of $C$ (else we truncate them) and that, as $B$ is 2-connected, they contain no common vertex other than $t$.  We can form two cycles that contain $t$ by adding to $P+P'$ each of the two paths between $u$ and $u'$ in $C$.  As $C$ is an odd cycle, the lengths of these two paths, and therefore the lengths of the two cycles, have distinct parity.  Thus $t$ belongs to an odd cycle. 
Finally we note that the checks of the block decomposition needed to decide whether or not $G$ is a $T$-forest or $T$-bipartite can be done in $O(n+m)$ time.
\qed
\end{proof}

\section{Subset Vertex Cover}\label{s-svc}

In this section we present some results on {\sc Subset Vertex Cover}, some of which we will need later on.

\begin{lemma}\label{svc-p4}
{\sc Subset Vertex Cover} can be solved in polynomial time for $P_4$-free graphs.
\end{lemma}

\begin{proof}
Let $G$ be a cograph with $n$ vertices and $m$ edges. First construct a modified cotree~$T_G'$ and then consider each node of $T_G'$  starting at the leaves of $T_G'$ and ending at the root $r$. Let $x$ be a node of $T_G'$. We let $S_x$ denote a minimum $(T\cap V(G_x))$-vertex cover of $G_x$. 

If $x$ is a leaf, then 
$G_x$ is a $1$-vertex graph. Hence, we can let $S_x=\emptyset$.
Now suppose that~$x$ is a $\oplus$-node.
Let $y$ and $z$ be the two children of $x$. Then, as $G_x$ is the disjoint union of $G_y$ and $G_z$, we can let $S_x=S_y\cup S_z$.
Finally suppose that  $x$ is a $\otimes$-node.
Let $y$ and $z$ be the two children of $x$. As $G_x$ is the join of $G_y$ and $G_z$ we observe the following:
if $V(G_x)\setminus S_x$ contains a vertex of $T\cap V(G_y)$, then $V(G_z)\subseteq S_x$. 
Similarly, if $V(G_x)\setminus S_x$ contains a vertex of $T\cap V(G_z)$, then $V(G_y)\subseteq S_x$. 
Hence, we  let $S_x$ be the smallest set of $S_y\cup V(G_z)$, $S_z\cup V(G_y)$ and 
$T\cap V(G_x)$.

Constructing $T_G'$ takes $O(n+m)$ time by Lemma~\ref{l-cotree}.
As~$T_{G'}$ has~$O(n)$ nodes and processing a node takes $O(1)$ time, the total running time is $O(n+m)$.
\qed
\end{proof}
The following lemma generalizes a corresponding well-known observation for {\sc Vertex Cover}.

\begin{lemma}\label{svc-ext}
Let $H$ be a graph. If {\sc Subset Vertex Cover} is polynomial-time solvable for $H$-free graphs, then it is for $(P_1+H)$-free graphs as well. 
\end{lemma}

\begin{proof} 
Let $G=(V,E)$ be a $(P_1+H)$-free graph and let $T\subseteq V$. Let $S_T$ be a minimum $T$-vertex cover of $G$.
For each vertex $u\in T$ we consider the option that $u$ belongs to $V\setminus S_T$. If so, then $N(u)$ belongs to $S_T$. Let $G'=G-N[u]$ and let 
$T'=T\setminus N[u]$. As $G'$ is $H$-free, we find a minimum $T'$-vertex cover $S_{T'}$ of $G'$ in polynomial time. We remember the smallest set $S_{T'}\cup N(u)$ and compare it with the size of $T$ to find $S_T$ (or some other minimum solution for $(G,T)$). \qed
\end{proof}
Lemma~\ref{svc-p4}, combined with $s$ applications of Lemma~\ref{svc-ext}, yields the following result.

\begin{theorem}\label{c-svc}
For every integer $s\geq 1$, {\sc Subset Vertex Cover} can be solved in polynomial time for $(sP_1+P_4)$-free graphs.
\end{theorem}

\section{Subset Feedback Vertex Set}\label{s-sfvs}

In this section we prove Theorem~\ref{t-main}.
Our contribution is \cref{sfvs-sp1p3}, which is the case where $H = sP_1+P_3$.
In the next section, we present an analogous result for {\sc Subset Odd Cycle Transversal}.  The proofs are similar in outline, but the latter requires additional insights.

We require  two lemmas.
In the first lemma, note that the bound of $4s-2$ is not necessarily tight, but is sufficient for our needs.

\begin{lemma}\label{tree-sp1p3}
Let $s$ be a non-negative integer, and let $R$ be an $(sP_1 + P_3)$-free tree.  Then either
\begin{enumerate}
\item[\rm{(i)}] $|V(R)| \le \max\{7,4s-2\}$, or
\item[\rm{(ii)}] $R$ has precisely one vertex $r$ of degree more than~$2$ and at most $s-1$ vertices of degree~$2$, each adjacent to $r$;
moreover, $r$ has at least $3s-1$ neighbours.
\end{enumerate}
\end{lemma}

\begin{proof}
If $R$ has no vertices of degree more than~$2$, then $R$ is a path, so $|V(R)|\leq 2s+2\leq\max\{7,4s-2\}$, otherwise $R$ has an induced $sP_1+P_3$ subgraph. Now let $r$ be a vertex of degree more than~$2$, and let $x$, $y$ and $z$ be distinct neighbours of $r$. We view $r$ as the root of the tree, and for $v\in V(R)$ we use $T_v$ to denote the subtree rooted at~$v$.

Suppose that $T_x$ has a vertex of degree at least~$2$. Then $T_x$ has an induced $P_3$, so $R - (V(T_x) \cup \{r\})$ is $sP_1$-free. As a tree is bipartite, this means that
this subtree consists of at most $2(s-1)$ vertices. Likewise, $R[\{y,r,z\}] \cong P_3$, so $T_x -x$ is $sP_1$-free, and hence consists of at most $2(s-1)$ vertices. Thus $|V(R)| \le 2(s-1) + 2(s-1) + 2 = 4s-2$.

We may now assume that for each $v \in N(r)$, the subtree $T_v$ has no vertices of degree at least~$2$; that is, either $T_v \cong P_1$ or $T_v \cong P_2$. It remains to show that when (i) does not hold, at most $s-1$ of the $T_v$ subgraphs are isomorphic to $P_2$. Towards a contradiction, suppose that $R$ has $s$ vertices at distance~$2$ from $r$, and $|V(R)| > \max\{7,4s-2\}$. Since $|V(R)| > 2(s+1) + 1$ for any non-negative integer $s$, the vertex $r$ has at least $s+2$ neighbours. Without loss of generality, label the neighbours of $r$ as $v_1, v_2, \dotsc, v_{deg(r)}$ such that $T_{v_i} \cong P_2$ for each $i \in \{1,\dotsc,s\}$.  Then $R[v_{s+1},r,v_{s+2}] \cong P_3$, and $T_{v_i} - \{v_i\} \cong P_1$ for each $i \in \{1,\dotsc,s\}$; a contradiction.

Finally, $|N_R(r)|+(s-1)+1\geq |V(R)|\geq 4s-1$, so $|N_R(r)|\geq 3s-1$. \qed
\end{proof}

\begin{figure}
\begin{center}
\begin{tikzpicture}[scale=0.7]
\filldraw [black] (-1,0) circle [radius=3pt] (1,3) circle [radius=3pt] (1,2) circle [radius=3pt] (1,1) circle [radius=3pt]
(1,0) circle [radius=3pt] (1,-3) circle [radius=3pt] (1,-0.6) circle [radius=2pt] (1,-1.2) circle [radius=2pt] (1,-1.8) circle [radius=2pt]
(1,-2.4) circle [radius=2pt] (2.5,3) circle [radius=3pt] (2.5,2) circle [radius=3pt] (2.5,1) circle [radius=3pt];
\draw (-1,0)--(1,3)--(2.5,3) (-1,0)--(1,2)--(2.5,2) (-1,0)--(1,1)--(2.5,1) (-1,0)--(1,0) (-1,0)--(1,-0.6) (-1,0)--(1,-1.2) (-1,0)--(1,-1.8) (-1,0)--(1,-2.4) (-1,0)--(1,-3)
(2.7,3.5)--(3,3.5)--(3,0.5)--(2.7,0.5) (1.2,3.5)--(1.5,3.5)--(1.5,-3.5)--(1.2,-3.5);
\node[left] at (-1,0) {$r$};
\node[right] at (3,2) {$\leq s-1$};
\node[right] at (1.5,-1.5) {$\geq 3s-1$};
\end{tikzpicture}
\caption{The structure of an $(sP_1+P_3)$-free tree, as given by Lemma~\ref{tree-sp1p3}, when (i) does not hold.}
\label{figbigtree}
\end{center}
\end{figure}
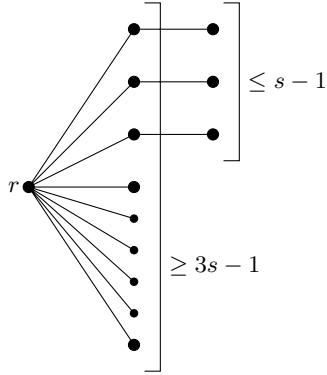

\noindent
Recall that for a graph $G=(V,E)$ and a set $S\subseteq V$, we let $G-S$ denote the graph $G[V\setminus S]$.
We can extend ``partial'' solutions to full solutions in polynomial time as follows.

\begin{lemma}\label{sfvs-solvep3free}
Let $G=(V,E)$ be a graph with a set 
$T \subseteq V$.
Let $V' \subseteq V$ and $S'_T \subseteq V'$ such that $S'_T$ is a $T$-feedback vertex set of $G[V']$, and let $Z = V \setminus V'$.
Suppose that $G[Z]$ is $P_3$-free, and $|N_{G-S'_T}(Z)| \le 1$.
Then there is a polynomial-time algorithm that finds a minimum $T$-feedback vertex set $S_T$ of $G$ such that $S'_T \subseteq S_T$ and $V' \setminus S'_T \subseteq V \setminus S_T$.
\end{lemma}

\begin{proof}
Since $G[Z]$ is $P_3$-free, it is a disjoint union of complete graphs. Let $G' = G-S'_T$. 
Suppose that $C$ is a $T$-cycle in $G'$. 
Then $C$ contains at least one vertex of $Z$. If $N_{G'}(Z) = \emptyset$, then $C$ is contained in a component of $G[Z]$. On the other hand, if $N_{G'}(Z) = \{y\}$, say,
then $y$ is a cutvertex of $G'$, so there exists a component $G[U]$ of $G[Z]$ such that $C$ is contained in $G[U \cup \{y\}]$. Hence, we can consider each component of $G[Z]$ independently: for each component $G[U]$ it suffices to find the maximum subset $U'$ of $U$ such that $G[U'\cup N_{G'}(U)]$ contains no $T$-cycles. Then $U' \subseteq F_T$ and $U \setminus U' \subseteq S_T$, where $F_T = V \setminus S_T$. 
So, $S_T$ will be the union of $S'_T$  and the vertex sets $U\setminus U'$, for every component $G[U]$ of $G[Z]$. Hence, it remains to prove how to find the sets $U'$ in polynomial time; we show this below.

Let $U \subseteq Z$ such that $G[U]$ is a component of $G[Z]$. Either $N_{G'}(U) \cap T = \emptyset$, or $N_{G'}(U) = \{y\}$ for some $y \in T$. First, consider the case where $N_{G'}(U) \cap T = \emptyset$. We find a set $U'$ that is a maximum subset of $U$ such that $G[U'\cup N_{G'}(U)]$ has no $T$-cycles. Clearly if $|U|=1$, then we can set $U'=U$. If $|U'| \ge 3$, then, since $U'$ is a clique, $U' \subseteq V \setminus T$. Thus, if $|U\setminus T|\geq 2$, then we set $U' = U\setminus T$. So it remains to consider when $|U|\geq 2$ but $|U\setminus T|\leq 1$. If there is some $u \in U$ that is anti-complete to $N_{G'}(U)$, then we can set $U'$ to be any $2$-element subset of $U$ containing $u$. Otherwise $N_{G'}(U) = \{y\}$ and $y$ is complete to $U$. In this case, for any $u \in U$, we set $U'=\{u\}$.

Now we may assume that $N_{G'}(U) = \{y\}$ and $y \in T$. Again, we find a set $U'$ that is a maximum subset of $U$ such that $G[U'\cup \{y\}]$ has no $T$-cycles. Partition $U$ into $\{U_0,U_1\}$ where $u \in U_1$ if and only if $u$ is a neighbour of $y$. Since $y \in V' \setminus S_T'$, observe that $U'$ contains at most one vertex of $U_1$, otherwise $G[U' \cup \{y\}]$ has a $T$-cycle. Since $U'$ is a clique, if $|U'| \ge 3$ then $U' \subseteq U \setminus T$. So if $|U_0\setminus T|\geq 2$ and there is an element $u\in U_1\setminus T$, then we can set $U'=\{u\}\cup (U_0\setminus T)$. If $|U_0\setminus T|\geq 2$ but $U_1\setminus T=\emptyset$, then we can set $U'=U_0\setminus T$. So we may now assume that $|U_0\setminus T|\leq 1$. If $U_0\neq \emptyset$ and $|U|\geq 2$, then we set $U'$ to any $2$-element subset of $U$ containing some $u \in U_0$. Clearly if $|U| = 1$, then we can set $U'=U$. So it remains to consider when $U_0 = \emptyset$ and $|U_1| \ge 2$.  In this case, we set $U' = \{u\}$ for an arbitrary $u \in U_1$. \qed
\end{proof}

We now prove the main result of this section.

\begin{theorem}\label{sfvs-sp1p3}
For every integer $s\geq 0$, {\sc Subset Feedback Vertex Set} can be solved in polynomial time for $(sP_1+P_3)$-free graphs.
\end{theorem}

\begin{proof}
Let $G=(V,E)$ be an $(sP_1+P_3)$-free graph for some $s\geq 0$, and let $T\subseteq V$.
We describe a polynomial-time algorithm for the optimization version of the problem on input $(G,T)$.
Let $S_T \subseteq V$ such that $S_T$ is a minimum $T$-feedback vertex set of $G$, and let $F_T = V \setminus S_T$, so $G[F_T]$ is a maximum $T$-forest.
Note that $G[F_T\cap T]$ is a forest.
We consider three cases: either 
\begin{enumerate}
\item $G[F_T \cap T]$ has at least $2s$ components; 
\item $G[F_T \cap T]$ has fewer than $2s$ components, and each of these components consists of at most $\max\{7,4s-2\}$ vertices; or
\item $G[F_T \cap T]$ has fewer than $2s$ components, one of which consists of at least $\max\{8,4s-1\}$ vertices.
\end{enumerate}
We describe polynomial-time subroutines that find a set $F_T$ such that $G[F_T]$ is a maximum 
$T$-forest in each of these three cases, giving a minimum solution $S_T = V \setminus F_T$ in each case.
We obtain an optimal solution by running each of these subroutines in turn: of the (at most) three potential solutions, we output the one with minimum size.

\smallskip
\noindent
{\bf Case 1:} $G[F_T\cap T]$ has at least $2s$ components. 

\smallskip
\noindent
We begin by proving a sequence of claims that describe properties of a maximum 
$T$-forest~$F_T$, when in Case~1.
Since $G$ is $(sP_1+P_3)$-free, $F_T\cap T$ induces a $P_3$-free forest, so $G[F_T \cap T]$ is a disjoint union of graphs isomorphic to $P_1$ or $P_2$.
Let $A \subseteq F_T \cap T$ such that $G[A]$ consists of precisely $2s$ components.  Note that $|A| \le 4s$.
We also let $Y = N(A) \cap F_T$, and partition~$Y$ into $\{Y_1,Y_2\}$ where $y \in Y_1$ if $y$ has only one neighbour in $A$, whereas $y \in Y_2$ if $y$ has at least two neighbours in $A$.

\medskip
\noindent
{\it Claim~1: $|Y_2|\leq 1$. }

\smallskip
\noindent
{\it Proof of Claim~1.}
Let $v \in Y_2$.  Then $v$ has neighbours in at least $s+1$ of the components of~$G[A]$, otherwise $G[A \cup \{v\}]$ contains an induced $sP_1+P_3$.
Note also that $v$ has at most one neighbour in each component of $G[A]$, otherwise $G[F_T]$ has a $T$-cycle.
Now suppose that~$Y_2$ contains distinct vertices $v_1$ and $v_2$.
Then, of the $2s$ components of $G[A]$, the vertices $v_1$ and $v_2$ each have some neighbour in $s+1$ of these components.  So there are at least two components of $G[A]$ containing both a vertex adjacent to $v_1$, and a vertex adjacent to $v_2$. 
Let $A'$ and $A''$ be the vertex sets of two such components.
Then $A' \cup A'' \cup \{v_1,v_2\} \subseteq F_T$, but $G[A' \cup A'' \cup \{v_1,v_2\}]$ has a $T$-cycle; a contradiction.  \dia

\medskip
\noindent
{\it Claim~2: $|Y|\leq 2s+1$. }

\smallskip
\noindent
{\it Proof of Claim~2.}
By Claim~1, it suffices to prove that $|Y_1| \le 2s$.  We argue that each component of $G[A]$ has at most one neighbour in $Y_1$, implying that $|Y_1| \le 2s$.
Indeed, suppose that there is a component $G[C_A]$ of $G[A]$ having two neighbours in $Y_1$, say $u_1$ and~$u_2$.
Then $G[V(C_A) \cup \{u_1,u_2\}]$ contains an induced $P_3$ that is anti-complete to $A\setminus V(C_A)$, contradicting that $G$ is $(sP_1+P_3)$-free.  \dia

\medskip
\noindent
{\it Claim~3: $Y_1$ is independent, and no component of $G[A]$ of size~$2$ has a neighbour in~$Y_1$. }

\smallskip
\noindent
{\it Proof of Claim~3.}
Suppose that there are adjacent vertices $u_1$ and $u_2$ in $Y_1$.
Let $a_i$ be the unique neighbour of $u_i$ in $A$ for $i \in \{1,2\}$.
Note that $a_1 \neq a_2$, for otherwise $G[F_T]$ has a $T$-cycle.
Then $\{a_1,u_1,u_2\}$ induces a $P_3$, so $G[\{u_1,u_2\} \cup A]$ contains an induced $sP_1+P_3$, which is a contradiction.
We deduce that $Y_1$ is independent.

Now let $\{a_1,a_2\} \subseteq A$ such that $G[\{a_1,a_2\}]$ is a component of $G[A]$, and suppose that $u_1 \in Y_1$ is adjacent to $a_1$. Then $a_1$ is the unique neighbour of $u_1$ in $A$, so $G[\{u_1,a_1,a_2\}] \cong P_3$. Thus $G[\{u_1\}\cup A]$ contains an induced $sP_1+P_3$, which is a contradiction. \dia

\medskip
\noindent
{\it Claim~4: Let $Z = V \setminus N[A]$.  Then $N(Z) \cap F_T \subseteq Y_2$. }

\smallskip
\noindent
{\it Proof of Claim~4.}
Suppose that there exists $y \in Y_1$ that is adjacent to a vertex $c \in Z$.
Let $a$ be the unique neighbour of $y$ in $A$. Then $G[\{c,y\}\cup A]$ contains an induced $sP_1+P_3$, which is a contradiction.  So $Y_1$ is anti-complete to $Z$.
Now, if $c \in Z$ is adjacent to a vertex in $N[A] \cap F_T$, then $c$ is adjacent to $y_2$ where $Y_2 = \{y_2\}$. \dia

\begin{figure}
\begin{center}
\begin{tikzpicture}[scale=0.6]
\draw[dashed] 
(4.5,-0.84)--(0,2)--(4.5,4.84) 
(5.7,-2.02)--(0,2)--(4.1,-3.8); 
\filldraw [black] (-3,-4) circle [radius=3pt] (-3,-3) circle [radius=3pt] (-3,-2) circle [radius=3pt] (-3,-1)  circle [radius=3pt] 
(-3,0) circle [radius=3pt] (-3,1) circle [radius=3pt] (-3,2) circle [radius=3pt] (-3,3) circle [radius=3pt] (-3,4) circle [radius=3pt] 
(0,2) circle [radius=3pt] (0,0) circle [radius=3pt] (0,-1.5) circle [radius=3pt];
\draw (-3,1)--(0,2)--(-3,3) (-3,2)--(0,2) to[out=250,in=110] (0,-1.5)--(-3,-1) (-3,4)--(0,2)--(-3,0)--(0,0) (-3,-2)(-3,-3)--(-3,-4) (-3,-2)--(0,2)--(-3,-3)
(-3.2,-4.3)--(-3.5,-4.3)--(-3.5,4.3)--(-3.2,4.3) (0.2,0.3)--(0.5,0.3)--(0.5,-1.8)--(0.2,-1.8);
\draw[color=black, fill=gray!20] (5,2) ellipse (1.5cm and 3cm);
\draw[color=black, fill=gray!20] (5,-3) ellipse (1.2cm and 1.2cm);
\node[right] at (5.5,2) {$U$};
\node[right] at (0.5,-1) {$Y_1$};
\node[above] at (0,2.2) {$y_2$};
\node[left] at (-3.5,-0.5) {$A$};
\node[right] at (7,-1.5) {$Z$};
\draw[black] (3,5.5) -- (7,5.5) -- (7,-5) -- (3,-5) -- (3,5.5);
\end{tikzpicture}
\caption{An example of the structure obtained in Case~1 when $Y_2 = \{y_2\}$.}
\label{structurefig}
\end{center}
\end{figure}
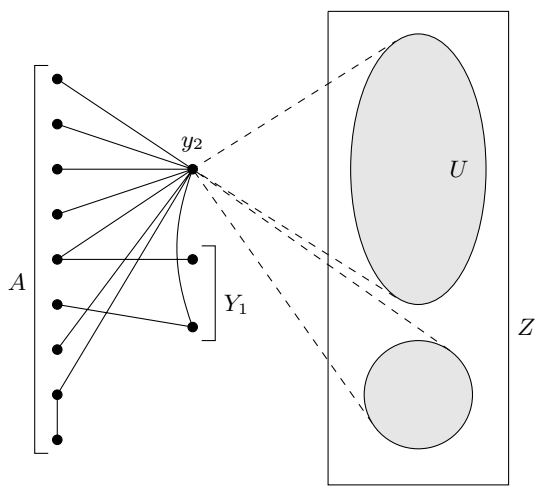

\noindent
We now describe the subroutine that finds an optimal solution in Case 1.
In this case, for any maximum forest $F_T$, there exists some set $A \subseteq T$ of size at most $4s$ such that $A \subseteq F_T$, and $G[A]$ consists of exactly $2s$ components, each isomorphic to either $P_1$ or $P_2$.
Since $G[A]$ consists of components of $G[F_T\cap T]$, there is such an $A$ for which $N(A) \cap T \subseteq S_T$.
Thus we guess a set $A' \subseteq T$ in $O(n^{4s})$ time, discarding those sets that do not induce a forest with exactly $2s$ components, and those that induce a component consisting of more than two vertices.

For any such $F_T$ and $A'$, the set $N(A') \cap F_T$ has size at most $2s+1$, by Claim~2.
Thus, in $O(n^{2s+1})$ time, we guess $Y' \subseteq N(A')$ with $|Y'| \le 2s+1$, and assume that $Y' \subseteq F_T$ whereas $N(A') \setminus Y' \subseteq S_T$.
Let $Y_2'$ be the subset of $Y'$ that contains vertices that have at least two neighbours in $A'$.
We discard any sets $Y'$ that do not satisfy Claims~1 or~3, or those sets for which $G[A'\cup Y']$ has a $T$-cycle on three vertices, one of which is the unique vertex of $Y_2'$. 

Let $Z=V\setminus N[A']$ (for example, see \cref{structurefig}).
Since $G[A']$ contains an induced $sP_1$, the subgraph $G[Z]$ is $P_3$-free.
Now $N(Z) \cap F_T \subseteq Y_2'$ by Claim~4, where $|Y_2'| \le 1$ by Claim~1.
Thus, by \cref{sfvs-solvep3free}, we can extend a partial solution $S'_T = N[A'] \setminus (A' \cup Y')$ of $G[N[A']]$ to a solution $S_T$ of $G$, in polynomial time.

\medskip
\noindent
{\bf Case 2:} $G[F_T\cap T]$ has at most $2s-1$ components, each of size at most $\max\{7,4s-2\}$.

\smallskip
\noindent
We guess sets $F \subseteq T$ and $S \subseteq V\setminus T$ such that $F_T \cap T = F$ and $S_T \setminus T = S$.
Since $F$ has size at most $(2s-1)\max\{7,4s-2\}$ vertices, there are 
$O(n^{\max\{14s-7,8s^2-8s+2\}})$ possibilities for $F$.
By Lemma~\ref{bound}, we may assume that $|S_T \setminus T| \le |F|$.
So for each guessed $F$, there are at most 
$O(n^{\max\{14s-7,8s^2-8s+2\}})$ possibilities for $S$.
For each $S$ and $F$, we set $S_T = (T \setminus F) \cup S$ and check, in $O(n+m)$-time by Lemma~\ref{st-test}, if $G-S_T$ is a $T$-forest. In this way we exhaustively find all solutions satisfying Case~2, in $O(n^{\max^2\{14s-7,8s^2-8s+2\}})$ time; we output the one of minimum size.

\medskip
\noindent
{\bf Case 3:} $G[F_T\cap T]$ has at most $2s-1$ components, one of which has size at least $\max\{8,4s-1\}$. 

\smallskip
\noindent
By Lemma~\ref{tree-sp1p3}, there is some subset $B_T \subseteq F_T\cap T$ such that $|B_T|\ge \max\{8,4s-1\}$, and $G[B_T]$ is a component of $G[F_T \cap T]$ that is a tree satisfying Lemma~\ref{tree-sp1p3}(ii), as illustrated in \cref{figbigtree}.
In particular, there is a unique vertex $r \in B_T$ such that $r$ has degree more than~$2$ in $G[B_T]$.
Moreover, $G[F_T]$ has a component $G[D]$ that contains $B_T$, where $G[D]$ is a tree that also satisfies \cref{tree-sp1p3}(ii).  Note that there are at most $s-1$ vertices in $N_{G[B_T]}(r)$ having a neighbour in $V \setminus T$.

We guess a set $B' \subseteq T$ such that $|B'| = \max\{8,4s-1\}$.
We also guess a set $L' \subseteq V \setminus T$ such that $|L'| \le s-1$.
Let $D' = B' \cup L'$.  We check that $G[D']$ has the following properties:

\begin{itemize}
\item $G[D']$ is a tree,
\item $G[D']$ has a unique vertex $r'$ of degree more than~$2$, with $r' \in B'$,
\item $G[D']$ has at most $s-1$ vertices with distance~$2$ from $r'$, and each of these vertices has degree~$1$, and
\item each vertex $v \in L'$ has degree~$1$ in $G[D']$, and distance~$2$ from $r'$.
\end{itemize}

We assume that $D'$ induces a subtree of the large component $G[D]$, where $r=r'$, and $D'$ contains~$r$, all neighbours of~$r$ with degree~$2$ in $G[D]$, and all vertices at distance~$2$ from~$r$.
In other words, $G[D']$ can be obtained from $G[D]$ by deleting some subset of the leaves of $G[D]$ that are adjacent to~$r$.
In particular, $D' \subseteq F_T$.
We also assume that $L'$ is the set of all vertices of $D \setminus T$ that have distance~$2$ from~$r$.

It follows from these assumptions that $N(D' \setminus \{r\}) \setminus \{r\} \subseteq S_T$.
Let $Z = V \setminus N[D' \setminus \{r\}]$, and observe that each $z \in Z$ has at most one neighbour in $D'$ (if it has such a neighbour, this neighbour is $r$).
So $N(Z) \cap F_T \subseteq \{r\}$.

In order to apply~\cref{sfvs-solvep3free}, it remains to show that $G[Z]$ is $P_3$-free.
Let $B_1 = B' \cap N(r)$.
As $r$ has at least $3s-1$ neighbours in $G[B']$, by Lemma~\ref{tree-sp1p3}, $G[B_1]$ contains an induced $sP_1$.
Moreover, $N(B_1) \cap F_T \subseteq D'$.
Since $G$ is $(sP_1+P_3)$-free, $G[Z]$ is $P_3$-free. 
Thus, by \cref{sfvs-solvep3free}, we can extend a partial solution $S'_T=N(D' \setminus \{r\}) \setminus \{r\}$ of $G[N[D' \setminus \{r\}]]$ to a solution $S_T$ of $G$, in polynomial time. \qed
\end{proof}

We are now ready to prove Theorem~\ref{t-main}.

\medskip
\noindent
{\bf Theorem~\ref{t-main} (restated).}
{\it Let $H$ be a graph with $H\neq sP_1+P_4$ for all $s\geq 1$. 
Then {\sc Subset Feedback Vertex Set} on $H$-free graphs is polynomial-time solvable if 
$H=P_4$ or $H\ssi sP_1+P_3$ for some $s\geq 1$, and is \NP-complete otherwise.}

\begin{proof}
If $H$ has a cycle or claw, we use Theorem~\ref{t-known}. The cases $H=P_4$ and $H=2P_2$ follow from the corresponding results for permutation graphs~\cite{PT19} and split graphs ~\cite{FHKPV14}. The remaining case $H\ssi sP_1+P_3$ follows from Theorem~\ref{sfvs-sp1p3}.\qed
\end{proof}

\section{Subset Odd Cycle Transversal}\label{s-soct}

At the end of this section we prove Theorem~\ref{t-main2}.  We need three new results to combine with existing knowledge. Our first result uses the reduction of~\cite{PT19} which proved the analogous result for {\sc Subset Feedback Vertex Set}. 

\begin{theorem}\label{soct-split}
{\sc Subset Odd Cycle Transversal} is \NP-complete for the class of split graphs (or equivalently, $(C_4,C_5,2P_2)$-free graphs).
\end{theorem}

\begin{proof}
We observe that the problem belongs to \NP~by Lemma~\ref{st-test}. To show \NP-hardness, we reduce from {\sc Vertex Cover}. 
Let a graph $G=(V,E)$ and a positive integer $k$ be an instance of {\sc Vertex Cover}. 
	From $G$, we construct a graph $G'$ as follows.
Let $V(G')=V\cup E$. Add an edge between $e\in E$ and $v\in V$ in $G'$ if and only if $v$ is an end-vertex of $e$ in $G$. 
Add edges so that~$V$ induces a clique of $G'$. Hence, $G'$ is a split graph with independent set $E$ and clique~$V$.
For example, when $G=P_4$, see \cref{constegfig}.
Let $T=E$. We show that $G$ has a vertex cover of size at most $k$ if and only if $G'$ has an odd $T$-cycle transversal of size at most~$k$.

First suppose that $G$ has a vertex cover $S$ of size at most $k$. Then $S$ is an odd $T$-cycle transversal of $G'$.
Now suppose that $G'$ has an odd $T$-cycle transversal~$S_T$ of size at most~$k$. 
As every vertex of $E$ in $G'$ has degree~$2$, we can replace every vertex of $E$ that belongs to $S_T$ by one of its neighbours to obtain an odd $T$-cycle transversal of the same size as $S_T$. Hence we may assume, without loss of generality, that $S_T\cap E=\emptyset$.
As a vertex of $E$ and its two neighbours in~$V$ form a triangle, this means that $S_T$ contains at least one neighbour of every $e\in E$. Hence, $S_T$ is a vertex cover of $G$.
\qed
\end{proof}

\begin{figure}
\begin{center}
\begin{tikzpicture}[xscale=0.5, yscale=0.5]
\draw (-4,3.3)--(-4.2,3.3)--(-4.2,-3.3)--(-4,-3.3) (2.4,2.3)--(2.6,2.3)--(2.6,-2.3)--(2.4,-2.3)
(-3,3)--(-3,-3)--(2,-2)--(-3,-1)--(2,0)--(-3,1)--(2,2)--(-3,3) (-3,-3) to[out=110,in=250] (-3,1) 
(-3,-1) to[out=110,in=250] (-3,3) (-3,-3) to[out=115,in=245] (-3,3);
\draw[fill=black] (-3,3) circle [radius=4pt] (-3,1) circle [radius=4pt] (-3,-1) circle [radius=4pt] (-3,-3) circle [radius=4pt]
(2,-2) circle [radius=4pt] (2,0) circle [radius=4pt] (2,2) circle [radius=4pt];
\node[left] at (-4.2,0) {$V$};
\node[right] at (2.7,0) {$E$};
\end{tikzpicture}
\caption{The graph $P_4'$: an example of the construction in the proof of \cref{soct-split}.}
\label{constegfig}
\end{center}
\end{figure}
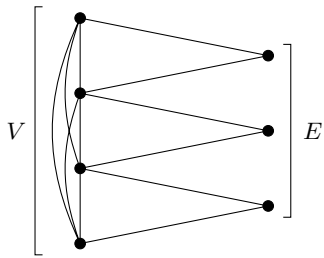

\begin{theorem}\label{soct-p4}
{\sc Subset Odd Cycle Transversal} can be solved in polynomial time for $P_4$-free graphs.
\end{theorem}

\begin{proof}
Let $G$ be a cograph with $n$ vertices and $m$ edges, and let $T\subseteq V(G)$. First construct the modified cotree~$T_G'$ and then consider each node of $T_G'$ starting at the leaves of $T_G'$ and ending in its root $r$. Let~$x$ be a node of $T_G'$. We let $S_x$ denote a minimum odd $(T\cap V(G_x))$-cycle transversal of $G_x$. 

If $x$ is a leaf, then $G_x$ is a 1-vertex graph. Hence, we can let $S_x=\emptyset$.
Now suppose that~$x$ is a $\oplus$-node.
Let $y$ and $z$ be the two children of $x$. Then, as $G_x$ is the disjoint union of $G_y$ and $G_z$, we let $S_x=S_y\cup S_z$.

Finally suppose that $x$ is a $\otimes$-node.
Let $y$ and $z$ be the two children of $x$. Let $T_y=T\cap V(G_y)$ and $T_z=T\cap V(G_z)$.
Let $B_x=V(G_x)\setminus S_x$. 
As $G_x$ is the join of $G_y$ and $G_z$ we observe the following.
If $B_x\cap V(G_y)$ contains two adjacent vertices, at least one of which belongs to~$T_x$, then $B_x\cap V(G_z)=\emptyset$ (as otherwise $G[B_x]$ has a triangle containing a vertex of $T$) and thus 
$V(G_z)\subseteq S_x$. 
In this case we may assume that $S_x=S_y\cup V(G_z)$.
Similarly, if $B_x\cap V(G_z)$ contains two adjacent vertices, at least one of which belongs to $T_z$, then $B_x\cap V(G_y)=\emptyset$ and thus $V(G_y)\subseteq S_x$. 
In this case we may assume that $S_x=S_z\cup V(G_y)$.
From the two sets $S_y\cup V(G_z)$ and $S_z\cup V(G_y)$ we remember the smallest one.
 
It remains to examine the case where $B_x\cap V(G_y)$ and $B_x\cap V(G_z)$ induce subgraphs of~$G$ in which the vertices of $T_y\cap B_x$ and $T_z\cap B_x$, respectively, are singleton components.

First suppose that  $T_y\cap B_x$ and $T_z\cap B_x$ are both non-empty.
Then $B_x\cap V(G_y)$ and $B_x\cap V(G_z)$ are both independent sets, as otherwise $G[B_x]$ would contain a $T$-triangle.
We examine this situation by computing a largest independent set $I_y$ in $G_y$ and a largest independent set $I_z$ in $G_z$; it is well-known that this can be done in polynomial time (for example, it follows from  Lemma~\ref{svc-p4}).
We remember $V(G_x)\setminus (I_y\cup I_z)$. {Since the union of two independent sets induces a bipartite graph,
this is an odd $T$-cycle transversal}.

Now suppose that  $T_y\cap B_x$ is non-empty, but $T_z\cap B_x$ is empty.
Then $B_x\cap V(G_z)$ must be an independent set, as otherwise we obtain a $T$-triangle by taking a vertex of $T_y\cap B_x$ and two adjacent vertices of $B_x\cap V(G_z)$.
First assume that $B_x\cap V(G_z)$ has size at least~$2$.
We observe that $(B_x\cap V(G_y))\setminus T_y$ is also an independent set; otherwise two adjacent vertices of $(B_x\cap V(G_y))\setminus T_y$, two vertices of $B_x\cap V(G_z)$ and one vertex of $T_y\cap B_x$ would form a $T$-cycle on five vertices. Hence, both $B_x\cap V(G_y)$ and $B_z\cap V(G_z)$ are independent sets, and we already dealt with this case above.

Now assume that $B_x\cap V(G_z)$ has size at most~$1$. 
In this case $B_x\cap V(G_y)$ is a minimum $T_y$-vertex cover of $G_y$. 
We can compute a minimum $T_y$-vertex cover $S$ of $G_y$ in polynomial time by {Lemma~\ref{svc-p4}}. 
We remember $S\cup (V(G_z)\setminus \{{z'}\})$ where ${z'}$ is an arbitrary vertex of $V(G_z)\setminus T_z$ if the latter set is non-empty; otherwise we just remember $S\cup V(G_z)$. 

We deal with the case where $T_z\cap B_x$ is non-empty, but $T_y\cap B_x$ is empty in the same way and remember the output. 
We also consider the possible situation where $T_z\cap B_x=T_y\cap B_x=\emptyset$, in which case we remember 
{$T \cap V(G_x)$}.
Finally, we take as set $S_x$ a set of minimum size over the sets that we remembered.

Constructing $T_G'$ takes $O(n+m)$ time by Lemma~\ref{l-cotree}.
As~$T_{G'}$ has~$O(n)$ nodes and processing a node takes $O(n+m)$ time (due to the application of Lemma~\ref{svc-p4}), the total running time is $O(n^2+mn)$.
\qed
\end{proof}

\noindent
The following result is the main result of this section. 
Its proof uses the same approach as the proof of Theorem~\ref{sfvs-sp1p3} but we need more advanced arguments. 

\begin{theorem} \label{soct-sp1p3}
For every integer $s\geq 0$, {\sc Subset Odd Cycle Transversal} can be solved in polynomial time for $(sP_1+P_3)$-free graphs.
\end{theorem}

\begin{proof}
Let $G=(V,E)$ be an $(sP_1+P_3)$-free graph and let $T\subseteq V$.
{Note that it suffices to prove the result for $s \ge 3$, since for every $0\leq s'\leq s$,
the class of  $(s'P_1+P_3)$-free graphs is contained in the class of $(sP_1+P_3)$-free graphs.}
We describe a polynomial-time algorithm to solve the optimization problem on input $(G,T)$.
That is, we describe how to find a smallest odd $T$-cycle transversal~$S_T$. In fact, we will solve the equivalent problem of finding a maximum size 
{set $B_T$ such that $G[B_T]$ is $T$-bipartite; so, of course, $B_T$ is} the complement of a smallest odd $T$-cycle transversal, that is $S_T=V\setminus B_T$.
We separate into two cases that separately seek to find $T$-bipartite subgraphs with complementary constraints on the size of the intersection of this subgraph with~$T$.  
The largest one found overall is the desired output.

\medskip
\noindent
{\bf Case~1:} 
Compute a largest $T$-bipartite subgraph {$G[B_T]$} such that $|B_T \cap T|\leq {4s-3}$. 

\smallskip
\noindent
Note that $B^*=V \setminus T$ is a candidate solution.  We must see if we can find something larger.  Consider each set $B' \subseteq T$ of size at most {$4s-3$}, discarding any set that does not induce a bipartite graph.  There are {$O(n^{4s-3})$} possible sets.  For each choice of $B'$, consider all sets $S \subseteq V \setminus T$ of size less than $|B'|$.
Then $B' \cup {(}(V \setminus T) \setminus S{)}$ is a candidate solution if it induces a $T$-bipartite subgraph, which is checked in $O(n+m)$-time by Lemma~\ref{st-test}.  For each $B'$, there are {$O(n^{4s-3})$} possible choices of $S$ to consider.  Note that we do not need to examine larger {sets} $S$
since then $B' \cup {(}(V \setminus T) \setminus S{)}$ is no larger than $B^*$.

\medskip
\noindent
{\bf Case~2:} 
Compute a largest $T$-bipartite subgraph {$G[B_T]$} such that $|B_T \cap T|\geq {4s-2}$.

\smallskip
\noindent
Note that $B_T$ might not exist in which case the output of Case 1 is our result.  We make some observations about the subgraph 
{$G[B_T]$} that we seek.  As $G[B_T\cap T]$ is a bipartite graph on at least {$4s-2$} vertices, it contains an independent set $A$ of size {$2s-1$}. Let $Y=B_T \cap N(A)$ and consider a partition $\{Y_1,Y_2\}$ of $Y$ where $y$ is in  $Y_1$ if~$y$ has precisely one neighbour in $A$, and otherwise $y$ is in~$Y_2$.  

\medskip
\noindent
{\it Claim~1: $Y_1$ is an independent set, no two vertices of $Y_1$ have a common neighbour in $A$ and $|Y_1| \le |A|$.}

\smallskip
\noindent
{\it Proof of Claim~1.}
Suppose that there are adjacent vertices $y,y'\in Y_1$, and let $a$ be the unique neighbour of $y$ in $A$. Then, according to whether or not $y'$ is adjacent to $a$, either $\{y,y',a\}$ induces an odd $T$-cycle, or $G[A \cup \{y,y'\}]$ contains an induced $sP_1+P_3$ {(recall that $s\geq 3$)}. Both are contradictions.  If there are {two (non-adjacent)}  vertices $y,y'\in Y_1$ that have the same neighbour $a$ in $A$, then, again, $G[A \cup \{y,y'\}]$ contains an induced $sP_1+P_3$, a contradiction. It follows {that no two vertices of $Y_1$ have a common neighbour in $A$, and hence} $|Y_1| \le |A|$.~\dia

\medskip
\noindent
{\it Claim~2: $Y_2$ is an independent set, each $y \in Y_2$ has at least $s$ neighbours in $A$ and any two vertices of $Y_2$ share at least one neighbour in $A$.}

\smallskip
\noindent
{\it Proof of Claim~2.}
{Let $y$ and $y'$ be two distinct vertices in $Y_2$. 
By definition, $y$ and $y'$ each have at least two neighbours in $A$}. 
  Since $G[A \cup \{y\}]$ is $(sP_1+P_3)$-free, $y$ is non-adjacent to at most $s-1$ vertices of $A$.  So $y$ has at least $2s-1-(s-1)=s$ neighbours in $A$. 
{Similarly, $y'$ has at least $s$ neighbours in $A$.
As $|A|=2s-1$, this means that $y$ and $y'$ have a common neighbour $a\in A$. If $y$ and $y'$ are adjacent, then  $\{y,y',a\}$  would induce an odd $T$-cycle, a contradiction.} \dia

\medskip
\noindent
Armed with these definitions and claims{,} we
{show} 
how to find $B_T$ {in polynomial time}.  
We have two subcases.

\medskip
\noindent
{\bf Case~2a:} Compute a largest $T$-bipartite subgraph {$G[B_T]$} such that $|B_T \cap T|\geq {4s-2}$ and for some choice of $A$, we have 
$|Y|\leq {3s}$.

\smallskip
\noindent
Consider each set $A \subseteq T$ of size {$2s-1$} such that $A$ is an independent set. There are {$O(n^{2s-1})$} choices 
{for $A$}.
For each $A$, we consider each set $Y_1 \subseteq N(A)$ of vertices that each have a single neighbour in $A$ such that~$Y_1$ satisfies Claim~1.  As we require that $|Y_1| \leq |A|$, there are  {$O(n^{2s-1})$} choices 
{for $Y_1$}.  
Then consider each set $Y \subseteq N(A)$ of size at most 
{$3s$}
 such that $Y_1 \subseteq Y$ and $Y_2 = Y \setminus Y_1$ is a set of vertices that each have at least two neighbours in $A$ and satisfies Claim~2.  We also require that 
$A \cup Y$
does not {induce} any 
odd $T$-cycles, which is checked in $O(n+m)$-time by Lemma~\ref{st-test}.  There are 
 {$O(n^{3s})$} choices for $Y$.

Note that $G[A  \cup Y]$ is bipartite since $G[Y]$ can contain only even cycles as~$Y_1$ and $Y_2$ are independent sets,
and any odd cycle in 
{$G[A \cup Y]$} 
{must thus contain a vertex of $A$ and so}
is an odd $T$-cycle, since $A\subseteq T$.
By Claim~2, vertices of $Y_2$ all belong to the same component of $G[A \cup Y]$ and, as, by definition and Claim~1, each vertex in $G[A \cup Y_1]$ has degree at most 1, we deduce that every vertex of degree at least~$2$ in $G[A \cup Y]$ belongs to the same component.  
We denote this component by $G[D]$, or we let $D$ be the empty set if there is no such component (which only occurs when $Y_2 = \emptyset$).  
See \cref{figbipart} for an illustration.

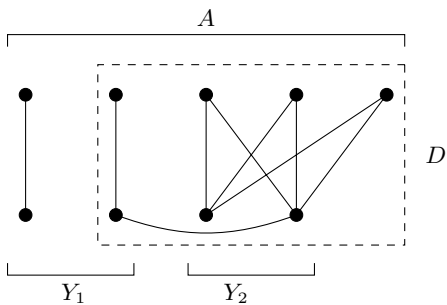
\begin{figure}
\centering
\begin{tikzpicture}[xscale=0.8, yscale=0.8]
\draw (1.5,0)--(3,2)--(0,0) (1.5,0)--(1.5,2)--(0,0) (1.5,0)--(0,2)--(0,0) (-1.5,2)--(-1.5,0) (-3,2)--(-3,0) (1.5,0) to[out=200,in=340] (-1.5,0) (3.3,2.8)--(3.3,3)--(-3.3,3)--(-3.3,2.8) (1.8,-0.8)--(1.8,-1)--(-0.3,-1)--(-0.3,-0.8) (-3.3,-0.8)--(-3.3,-1)--(-1.2,-1)--(-1.2,-0.8); \draw[dashed] (3.3,2.5)--(3.3,-0.5)--(-1.8,-0.5)--(-1.8,2.5)--(3.3,2.5); \draw[color=black, fill=black] (-3,0) circle [radius=3pt] (3,2) circle [radius=3pt] (1.5,2) circle [radius=3pt] (0,2) circle [radius=3pt] (-1.5,2) circle [radius=3pt] (-3,2) circle [radius=3pt] (0,0) circle [radius=3pt] (1.5,0) circle [radius=3pt] (-1.5,0) circle [radius=3pt];\node[above] at (0,3) {$A$};
\node[below] at (0.5,-1) {$Y_2$};
\node[below] at (-2.2,-1) {$Y_1$};
\node[right] at (3.5,1) {$D$};
\end{tikzpicture}
\caption{An example of $G[A\cup Y]$ in Case 2a when $s=3$.}
\label{figbipart}
\end{figure}

\noindent
We now let $Z = V \setminus N[A]$.

\medskip
\noindent
{\it Claim~3:   $N(Z) \cap B_T \subseteq Y_2$.}

\smallskip
\noindent
{\it Proof of Claim~3.}
By definition, $N(Z) \cap B_T \subseteq Y$. Suppose that $z \in Z$ is adjacent to a vertex $y \in Y_1$. Let $a$ be the unique neighbour of $y$ in $A$. Since $|A| = {2s-1} \ge s+1$ {as $s \ge 3$}, it follows that $G[\{z,y\}\cup A]$ contains an induced $sP_1+P_3$, a contradiction.  So $Y_1$ is anti-complete to~$Z$, and the claim follows. \dia

\medskip
\noindent
The aim in the remainder of this subcase is to find the largest possible 
set $B_T$ that induces a $T$-bipartite subgraph and 
that contains $A \cup Y$ and a subset of $Z$.
Since {$G[A]$} contains an induced $sP_1$ subgraph, $G[Z]$ is $P_3$-free, and so is a disjoint union of complete graphs. 
{From Claim~3, it follows that $(A \cup Y) \setminus D$ 
(which is a subset of $A\cup Y_1$)
 is anti-complete to $B_T \setminus (A \cup Y)$. Therefore, we only need to consider odd $T$-cycles in 
 $G[Z \cup D]$ 
 from now on.}
 
For a component $G[U]$ of $G[Z]$,  let $U^+$ be the subset of $U$ that consists of each vertex~$u$ of $U$ such that $G[A \cup Y \cup \{u\}]$ does not contain an odd $T$-cycle through $u$, which is checked in $O(n+m)$-time by Lemma~\ref{st-test}.  
Clearly for each component $G[U]$ in $G[Z]$, any vertex that might be in~$B_T$ must belong to $U^+$.  We shall see later that we can consider each component of $G[Z]$ independently and that it suffices to find for each the maximum size subset of $U^+$ that can be added to $B_T$.
{We note that by Claim~3, the neighbours of $U^+$ in $A \cup Y$ belong to $Y_2$ and, 
as $Y_2\subseteq D$, these neighbours therefore also belong
to $D$.  We first investigate the possible edges between $U^+$ and~$D$.}

\medskip
\noindent
{\it Claim~4:  Either $|N(U^+ \cap B_T) \cap {D}| \le 1$ or $|N(D)\cap (U^+ \cap B_T)| \le 1$.} 

\smallskip
\noindent
{\it Proof of Claim~4.}
We can assume that there are two {distinct} vertices $u_1, u_2$ of $U^+ \cap B_T$ that 
have a
{different} neighbour in $D$ else the claim follows immediately.  
{Let $y_1$ be a neighbour of $u_1$ in $D$ and $y_2$ be a neighbour of $u_2$ in $D$ with $y_1\neq y_2$.}
By Claim~2, $y_1$ and~$y_2$ have a common neighbour $a$ in $A$.  Thus we have a path $u_1y_1ay_2u_2$.  
{Recall that $G[Z]$ is a disjoint union of complete graphs. Hence,}
$U^+$ is a clique, {which means that the $5$-vertex path $u_1y_1ay_2u_2$} can be extended to a cycle by the edge $u_1u_2$, but, as $A \subseteq T$, this is an odd $T$-cycle, a contradiction.~\dia

\medskip
\noindent
{\it Claim~5: 
For each component $G[U]$ of $G[Z]$, let $U^{++}$ be a subset of $U^+$; and
let $Z^{++}$ be the union of each $U^{++}$ over all components $G[U]$ of $G[Z]$.
If $G[A\cup Y\cup Z^{++}]$ contains an odd $T$-cycle, then $G[A \cup Y \cup U^{++}]$ contains an odd $T$-cycle for some component $G[U]$ of $G[Z]$.}

\smallskip
\noindent
{\it Proof of Claim~5.}
Suppose that $C$ is an odd $T$-cycle of $G[A\cup Y\cup Z^{++}]$.
First we show that $C$ contains two vertices of some $U^{++}$.
Towards a contradiction, suppose $C$ is a subgraph of $G[A\cup Y\cup Z^*]$, where $Z^*$ is a subset of $Z^{++}$ with at most one vertex from each component of $G[Z]$.
Recall that $D$ is a bipartite graph that (if non-empty) is a component of $G[A \cup Y]$.
By Claim~3, all neighbours of $Z^*$ are contained in $Y_2$, which, in turn, is contained in one side of the bipartition of $D$.
Hence $G[A\cup Y\cup Z^*]$ {is bipartite and} has no odd $T$-cycles and, in particular, $C$ is not an odd $T$-cycle.
From this contradiction we deduce that there is some component $G[U]$ of $G[Z]$ such that $C$ contains two 
vertices of $U^{++}$.

{We can assume that $C$ is not contained in $G[A\cup Y\cup U^{++}]$ else the claim follows immediately.
Then there are distinct vertices $u_1\in U^{++}$ and $u_2\in U^{++}$ and distinct vertices  $y_1 \in N_C(u_1) \cap Y_2$ and $y_2 \in N_C(u_2) \cap Y_2$.}
By Claim~2, $y_1$ and $y_2$ have a common neighbour $a \in A$.
Then, 
{as $U$, and thus $U^{++}$, is a clique}, $u_1y_1ay_2u_2u_1$ is an odd $T$-cycle contained in $G[A \cup Y \cup U^{++}]$, 
{a contradiction}. \dia

\medskip
\noindent
{By Claim~5,}
to extend $A \cup Y$ to the largest possible $T$-bipartite graph, for each component~$G[U]$ of $G[Z]$, we must find 
{in polynomial time} 
a maximum subset $U^{++}$ of $U^{+}$ such that $G[A \cup Y \cup U^{++}]$ has no odd $T$-cycle.

{We find a set~$U^{++}$ as follows.}  
We first suppose that for the set we seek $|U^{++}| \ge 3$.
Note that in this case we have $U^{++}\cap T=\emptyset$, since $U^{++}$ is a clique.
Partition $U^{+} \setminus T$ into $\{U^{+}_0, U^{+}_1, U^{+}_2\}$ where $u\in U^{+}_0$ if $u \in U^{+} \setminus T$ has
no neighbours in ${D}$, $u\in U^{+}_1$ if $u$ has exactly one neighbour in ${D}$, and otherwise $u\in U^{+}_2$. 
If $U^{+}_1$ is not empty, then let $N(U^{+}_1) \cap {D} = \{d_1,\dotsc,d_m\}$, for some $m\geq 1$.
We partition $U^{+}_1$ into classes $\{Q_1,\dotsc,Q_m\}$ such that $u \in Q_i$ if $N(u) \cap {D} = \{d_i\}$.
Using Claim~4, either $U^{++} \cap U^{+}_1 = \emptyset$ or $U^{++} \cap U^{+}_2 = \emptyset$. Moreover, when $U^{++} \cap U^{+}_1 \neq \emptyset$, then $U^{++} \cap U^{+}_1 \subseteq Q_i$ for some $i \in \{1, \dotsc, m\}$; and when $U^{++} \cap U^{+}_2 \neq \emptyset$, then $|U^{++} \cap U^{+}_2| = 1$.
So, if there exists some $d_i \notin T$,
then we choose an 
{$i$ with $d_i\notin T$} 
that maximises $|Q_i|$, and set $U^{++} = U^{+}_0 \cup Q_i$.
If $U^{+}_1 \neq \emptyset$ but $d_i \in T$ for all $i \in \{1,\dotsc,m\}$, then set $U^{++} = U^{+}_0 \cup \{u\}$ for an arbitrarily chosen $u \in U^{+}_1$.
Now suppose $U^{+}_1$ is empty, and recall that in this case $|U^{+}_2 \cap U^{++}| \leq 1$.
If $U^{+}_2$ is non-empty, then set $U^{++} = U_0^{+} \cup \{u_2\}$ for some
{arbitrarily chosen}
 $u_2 \in U^{+}_2$.
Finally, if $U^+_2$ is also empty, then set $U^{++}=U_0^{+}$.
This process finds a maximum $U^{++}$ of size at least~$3$ if such a set exists.
{We note that $G[A \cup Y \cup U^{++}]$ does not contain an odd $T$-cycle since either the only neighbour of $U^{++}$ in $A \cup Y$ is some $d_i \not \in T$, or the only vertex in $U^{++}$ adjacent to $A \cup Y$ 
does not belong to $T$.}

Now consider the case where $|U^{++}| \le 2$.
We exhaustively check all pairs of vertices in $U^{+}$, of which there are $O(n^2)$.
Let $u_1,u_2$ be such a pair of distinct vertices.
We
check that $G[A \cup Y \cup \{u_1,u_2\}]$ is $T$-bipartite; if it is, then we set $U^{++}=\{u_1,u_2\}$.
Recall that this check runs in polynomial time, by \cref{st-test}.
If no pair is found, we set $U^{++}$ to be the singleton set consisting of any arbitrarily chosen vertex of $U^{+}$.
{If $U^+$ is empty, then we set $U^{++}=\emptyset$. 
We observe that in all cases we obtained $U^{++}$ in polynomial time, as needed.}

\medskip
\noindent
{\bf Case~2b:} Compute a largest $T$-bipartite subgraph {$G[B_T]$} such that $|B_T \cap T|\geq {4s-2}$ and for some choice of $A$, we have 
$|Y|\geq {3s+1}$.

\smallskip
\noindent
Note that as $A$ has size ${2s-1}$ and $|Y_1| \leq |A|$, we have that 
$|Y_2|\geq s+2$.
{Let $Y'_2$ be a subset} of $Y_2$ with $|Y'_2|=s+2$.  
Let $A_0=N(Y'_2)\cap A$.
{By Claim~2 it follows that $|A_0|\geq s$. As $A_0\subseteq A$, we also find that}
$|A_0|\leq {2s-1}$. 

\medskip
\noindent
{\it Claim~6: Let $y$ and $y'$ be distinct vertices of {$N(A_0)$}. Then there is an even $T$-path in 
$G[A_0\cup Y'_2 \cup \{y,y'\}]$ between $y$ and $y'$.
}

\smallskip
\noindent
{\it Proof of Claim~6.}
Assume that $y$ and $y'$ have no common neighbour in $A_0$ else the claim is immediate. 
First let us assume at least one 
{of}
$y$ and $y'$ is  in $Y'_2$. Without loss of generality, $y\in Y'_2$. 
{Then $y' \in N(A_0) \setminus Y'_2$ by the assumption that $y$ and $y'$ have no common neighbour and Claim~2.}
{Let $a' \in A_0$ be a neighbour of $y'$.  By definition, $a'$ has a neighbour $y''$ in $Y'_2$ and $y \neq y''$ else $a'$ is a common neighbour of $y$ and $y'$.  By Claim~2 again, $y$ and $y''$ have a common neighbour $a'' \in A_0$.}
Thus $ya''y''a'y'$ is an even $T$-path in $G[A_0\cup Y'_2 \cup \{y,y'\}]$ 
{(in particular, note that $a'$ and $a''$ belong to $T$, as $A_0\subseteq A$ belongs to $T$).}

Now we consider the case where $y,y'\in {N(A_0)}\setminus Y'_2$. Let $y^{*}$ be a vertex of $Y'_2$. By the previous case there are even $T$-paths in $G[A_0\cup Y'_2\cup \{y,y^{*}\}]$ between 
{$y$ and $y^{*}$ and between $y'$ and $y^{*}$. 
Consider the subgraph induced by the vertices of these two paths.
After removing the edge $yy'$ (if it exists) this subgraph is bipartite, since $G[A_0 \cup Y_2']$ is bipartite and the neighbours of $y$ and $y'$ are both in $A_0$.  We then discard edges from (even) cycles until a path between $y$ and $y'$ is obtained. This path is even, and moreover passes through $A_0\subseteq T$. Hence, it is an even $T$-path between $y$ and $y'$ in $G[A_0\cup Y_2'\cup \{y,y'\}]$.}\dia

\medskip
\noindent
{\it Claim~7: $N(A_0)\cap N(Y'_2)\cap B_T=\emptyset$, {or equivalently, $N(A_0)\cap N(Y'_2)\subseteq S_T$}.}

\smallskip
\noindent
{\it Proof of Claim~7.}
{For contradiction}, assume there is a vertex $v \in N(A_0)\cap N(Y'_2)\cap B_T$. By assumption there are vertices $a\in A_0\cap N(v)$ and $y\in Y'_2\cap N(v)$.
By definition of $A_0$, there is a vertex $y'\in Y'_2\cap N(a)$.  
{If $y'=y$ then $\{a,v,y\}$ induces an odd $T$-cycle, a contradiction.}
Suppose now that $y'\neq y$. Then by Claim~6 there is an even $T$-path~{$P$}
in
$G[A_0\cup Y'_2\cup \{y,y'\}]$ {$(= G[A_0\cup Y'_2])$} between $y$ and $y'$. 
{As $Y'_2\subseteq Y_2$ is an independent set by Claim~2 and $A_0$ is independent by definition,  
$v$ does not belong to $A_0\cup Y'_2$ and hence, $v$ is not on $P$.}
{If $a$ does not belong to $P$, then the cycle formed by $P$ and $yvay'$ {is} an odd $T$-cycle of $B_T$.
If $a$ does belong to $P$, then we can divide $P$ into a path $P_1$ from $y$ to $a$ and a path $P_2$ from $y'$ to $a$.  As $P_1$ and $P_2$ must have the same parity, one of the two cycles formed by $P_1$ and $avy$ and by $P_2$ and $ay'$ is an odd $T$-cycle of $B_T$.  In either case, we have a contradiction.}
\dia

\medskip
\noindent
By Claim~7, $N(A_0)\cap N(Y'_2)\subseteq S_T$.
We let $Y'_0= N(A_0) \setminus N(Y'_2)$. 
{Note that $Y_2'\subseteq B_T$ and $Y_2'\subseteq N(A_0)$. Hence, $Y_2'\subseteq Y_0'$. Note that $Y_0'$ contains no neighbours of $Y_2'$ by definition.
Moreover, $Y_2'\subseteq Y_2$, and $Y_2$ is an independent set by Claim~2. It follows that the vertices of $Y_2'$ are isolated in $G[Y_0']$.
Finally, as $|Y_2'|=s+2$, we observe that $G[Y_0'\setminus Y_2']$ is a disjoint union of complete graphs; see also \cref{setrelationships}.}

Let $Z=V \setminus N[A_0]$; {see also \cref{setrelationships}.}
{Recall that $A_0$ is a subset of $A$ of size at least~$s$} and that $A$ is an independent set. Hence, $G[A_0]$ has an induced $sP_1$.
Consequently, $G[Z]$ is $P_3$-free, so $G[Z]$ is a disjoint union of complete graphs. 
Let $G[U]$ be a component of $G[Z]$; 
{note that $U$ is a clique}.
Partition the vertices of $U$ into $\{U_0,U_1,U_2\}$, where $u\in U_0$ if $u$ has no neighbours in $Y'_0$, whereas $u\in U_1$ if all neighbours of $u$ in $Y'_0$ are in one component of $G[Y'_0]$, and otherwise $u\in U_2$ ({that is,} when $u$ has neighbours in distinct components of $G[Y'_0]$). 

\begin{figure}
\center
\begin{tikzpicture}[scale=0.8]
\draw[color=black, fill=black] (-3,4) circle [radius=2pt] (-3,3) circle [radius=2pt] (-3,2) circle [radius=2pt] (-3,0) circle [radius=2pt] (-3,-1) circle [radius=2pt] (-3,-2) circle [radius=2pt];
\draw (-3.5,4.5)--(-2.5,4.5)--(-2.5,1.5)--(-3.5,1.5)--(-3.5,4.5) (-3.5,0.5)--(-2.5,0.5)--(-2.5,-2.5)--(-3.5,-2.5)--(-3.5,0.5) (0,7)--(3,7)--(3,2)--(0,2)--(0,7) (5,3)--(8,3)--(8,-2)--(5,-2)--(5,3) (-3,1.5)--(-3,0.5) (-2.5,3)--(0,4) (-2.5,-1.5)--(5,0) (3,3.5)--(5,1.5);
\draw[dotted, very thick] (-2.5,-0.5)--(0,2.3) (-2.5,2)--(5,0.75);
\draw[color=black,fill=gray!30!white] (1.5,6) ellipse (0.5cm and 0.5cm);
\draw[color=black,fill=gray!30!white] (1.5,3) ellipse (0.5cm and 0.5cm);
\draw[color=black,fill=gray!30!white] (1.5,4.5) ellipse (0.5cm and 0.5cm);
\draw[color=black,fill=gray!30!white] (6.5,2) ellipse (0.5cm and 0.5cm);
\draw[color=black,fill=gray!30!white] (6.5,0.5) ellipse (0.5cm and 0.5cm);
\draw[color=black,fill=gray!30!white] (6.5,-1) ellipse (0.5cm and 0.5cm);
\node[above] at (-3,4.5) {$A_0$};
\node[right] at (-2.5,-2) {$Y'_2$};
\node[right] at (3,5) {$Y'_0\setminus Y'_2$};
\node[above] at (6.5,3) {$Z=V\setminus N[A_0]$};
\end{tikzpicture}
\caption{The decomposition of $G$ in Case~2b excluding the set $N(A_0)\cap N(Y'_2)$ (as this set will belong to~$S_T$). Recall that $A_0=N(Y_2')\cap A$ and $Y_2'$ are independent sets with $s\leq |A_0|\leq 2s-1$ and $|Y_2'|=s+2$.
Moreover, $Y_0'=N(A_0)\setminus N(Y_2')$, so $Y_0'\setminus Y_2'$ is anti-complete to $Y_2'\subseteq Y_0'$. The graphs $G[Y_0']$ and $G[Z]$ are disjoint unions of complete graphs (in particular, the vertices of 
$Y_2'$ 
are isolated in $G[Y_0']$). Finally, $A_0$ is anti-complete to $Z=V\setminus N[A_0]$ by definition.}
\label{setrelationships}
\end{figure}
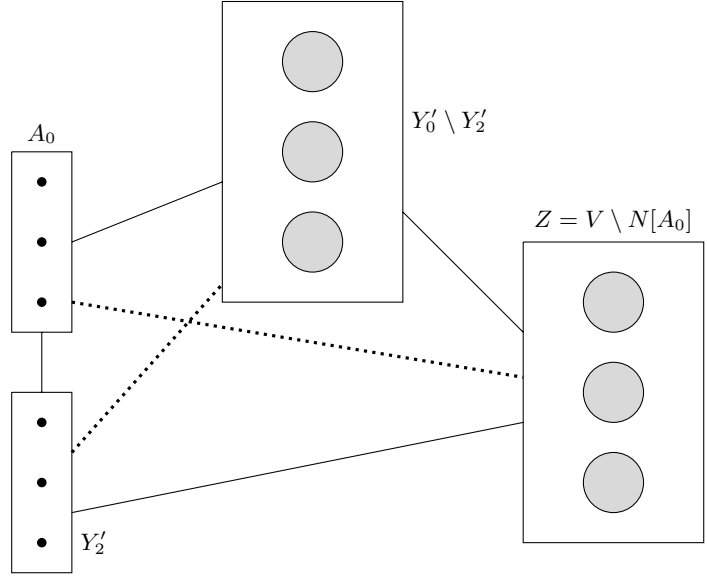

\medskip
\noindent
{\it Claim~8: If $u \in U_2$, then $u$ has at least two neighbours in $Y_2'$.}

\smallskip
\noindent
{\it Proof of Claim~8.}
Suppose that $u \in U_2$. 
{By definition, $u$ has neighbours $x$ and $y$ in distinct components of $G[Y'_0]$. By definition, $Y_0'$ contains no neighbours of $Y_2'$. Recall that the vertices of $Y_2'$ are isolated in $G[Y_0']$. Hence, $Y'_2\cup \{x,y\}$ is independent. As $G$ is $(sP_1+P_3)$-free and  $|Y_2'|=s+2$, the claim follows.} \dia

\medskip
\noindent
{{\it Claim~9:  Either $U_0 = \emptyset$ or $U_1 = \emptyset$. Moreover, if $U_1 \neq \emptyset$, then $U_1\cup (N(U_1) \cap Y'_0)$ is a clique.}}

\smallskip
\noindent
{\it Proof of Claim~9.}
{For contradiction, assume that $U_0$ and $U_1$ are both non-empty. Let $u_0\in U_0$, $u_1\in U_1$ and $y\in N(u_1)\cap Y'_0$.  Recall that $U$ is a clique, so $u_0$ and $u_1$ are adjacent and thus $\{u_0,u_1,y\}$ induces a $P_3$.
Recall also that $Y_0'$ contains no neighbours of $Y_2'$ and that the vertices of $Y_2'$ form an independent set of size $s+2\geq s$ in $G[Y_0']$. Then $G[\{u_0,u_1,y\}\cup Y_2']$ contains an induced $sP_1+P_3$; a contradiction.}

{Now let  $u_1$ and $u'_1$ be two distinct vertices in $U_1$. Let 
$y_1 \in N(u_1)\cap Y'_0$ and $y'_1 \in N(u'_1)\cap Y'_0$.  If $y_1$ is not adjacent to $u'_1$, then $\{y_1,u_1,u'_1\}$ induces a $P_3$. Then, by the same reason as above, $G[\{u_1,u_1',y_1\}\cup Y_2']$ contains an induced $sP_1+P_3$. Hence, $y_1$ is adjacent to $u_1'$, and by the same arguments, $y_1'$ is adjacent to $u_1$.
If $y_1 \neq y'_1$ and $y_1$ and $y'_1$ are not adjacent, then $\{y_1,u_1,y'_1\}$ induces a $P_3$ and thus
$G[\{u_1,y_1,y_1'\}\cup Y_2']$ contains an induced $sP_1+P_3$; a contradiction. We conclude that $U_1\cup (N(U_1) \cap Y'_0)$ is a clique.}
\dia

\medskip
\noindent
{\it Claim~10: $|U_2 \cap B_T| \le 1$.}

\smallskip
\noindent
{\it Proof of Claim~10.}
{For contradiction}, assume there exist {vertices} $u,u' \in U_2 \cap B_T$ with $u \neq u'$. By Claim~8, {we find that} $u$ and $u'$ each have at least two neighbours in $Y'_2$. Hence, there exist vertices $y,y' \in Y'_2$ such that $y \in N(u)$, $y' \in N(u')$ and $y \neq y'$. 
By Claim~6, there is an even $T$-path~$P$ in 
$G[A'_0 \cup Y_2']$ {$(=G[A_0'\cup Y_2' \cup \{y,y'\}])$} 
between $y$ and $y'$. {By} using the path $yuu'y'$, {we can extend} $P$  to an odd $T$-cycle; a contradiction. \dia

\medskip
\noindent
{
{\it Claim~11: Suppose that $u_1,u_2 \in B_T$ for some $u_1 \in U_1$ and $u_2 \in U_2$. Let $y$ be a vertex in $N(U_1) \cap Y'_0$, which is a clique. Then $y \in Y'_0 \setminus Y_2'$ and $y \notin B_T$.}}

\smallskip
\noindent
{\it Proof of Claim~11.}
{By Claim~8, we find that} $u_2$ has at least two neighbours in $Y'_2$. {Hence,} $u_2$ has a neighbour $y' \in Y'_2$ such that $y' \neq y$. By Claim~6, there is an even $T$-path~$P$ in $G[A_0 \cup Y_2' \cup \{y\}]$ between $y$ and $y'$.
{As, by Claim 9, $U_1\cup (N(U_1) \cap Y'_0)$ is a clique, $y$ and $u_1$ are adjacent. 
As $U$ is a clique, $u_1$ and $u_2$ are adjacent.}
{By} using the path $yu_1u_2y'$, {we can extend} $P$ to an odd $T$-cycle. Since $V(P) \setminus \{y\} \subseteq A_0 \cup Y_2' \subseteq B_T$ and $u_1,u_2 \in B_T$, {this implies} that $y \notin B_T$ {and thus $y\notin Y_2'$}. \dia

\medskip
\noindent
{\it Claim~12:
{$B_T$ contains at most one vertex of each component of $G[Y_0']$.}}

\medskip
\noindent
{\it Proof of Claim~12.}
{Recall that $G[Y_0']$ is the disjoint union of complete graphs.  For contradiction, assume that $y$ and $y'$ are vertices in $Y_0' \cap B_T\subseteq N(A_0)\cap B_T$  that belong to the same complete graph.  By Claim~6, there is an even $T$-path between $y$ and $y'$ that uses only vertices of $A_0\cup Y_2'\cup \{y,y'\} \subseteq B_T$.  This can be extended to an odd $T$-cycle by adding the edge between $y$ and $y'$, a contradiction.}  \dia  

\medskip
\noindent
{\it Claim~13:
{If the graph $G[A_0\cup (Y_0'\cap B_T)\cup Z]$  contains an odd $T$-cycle $C$, then there is a component~$G[U]$ of $G[Z]$ such that the graph $G[A_0\cup (Y_0'\cap B_T)\cup U]$ contains an odd $T$-cycle~$C'$.}}

\medskip
\noindent
{\it Proof of Claim~13.}
{Recall that $A_0$ is a subset of $B_T$.
Let $C$ be an odd $T$-cycle in $G[A_0\cup (Y_0'\cap B_T)\cup Z]$.
We can assume that $C$ contains vertices from at least two components of $G[Z]$ else we are done.
Since $B_T \subseteq A_0 \cup Y_0'\cup Z$, and $Z$ is anti-complete to $A_0$,
the cycle $C$ is the concatenation of a number of the following two types of paths on at least two vertices: 
a path is of \emph{type~1} if it starts and ends in $Y_0'\cap B_T$ and is contained in $G[A_0\cup (Y_0'\cap B_T)]$; 
a path is of \emph{type~2} if it starts and ends in $Y_0'\cap B_T$ and all its internal vertices are contained in a component of $G[Z]$.
Since $A_0$ and, by Claim 12, $Y_0' \cap B_T$ are independent sets, the graph $G[A_0\cup (Y_0' \cap B_T)]$ is bipartite and all the subpaths of $C$ of type~1 are even. Moreover, since $C$ is an odd cycle, there is a subpath~$P$ of $C$
of type~2 that is odd. Let $y$ and $y'$ be the two end-vertices of $P$, so $y$ and $y'$ both belong to $Y_0'\cap B_T$, which is a subset of $N(A_0)$.
By Claim~6, there exists an even $T$-path $Q$ in $G[A_0\cup Y_2'\cup \{y,y'\}]$, which is an induced subgraph of
$G[A_0\cup (Y_0'\cap B_T)]$
as $Y_2'\cup \{y,y'\}\subseteq Y_0'\cap B_T$.
Let $G[U]$ be the component of $G[Z]$ that contains the internal vertices of~$P$.
Joining $P$ and $Q$ yields the desired 
odd $T$-cycle~$C'$ in $G[A_0\cup (Y_0'\cap B_T)\cup U]$.}
\dia

\medskip
\noindent
{
Let us now describe how we find the largest possible $B_T$.
Our approach is to first guess a constant number of vertices to be in $B_T$. To be more precise, we consider each possible pair of sets $A_0$, with $s \leq |A_0| \leq 2s-1$, and $Y'_2$, of size $s+2$, that conform with the definitions of this subcase and Claim~6; note that the latter can be verified in polynomial time as the graph $G[A_0\cup Y_2'\cup \{y,y'\}]$ with $y,y'\in N(A_0)$ has constant size.  The total number of choices is $O(n^{(2s-1)(s+2)})$. For each choice of these two sets we want to find the largest possible $B_T$ that contains $A_0\cup Y_2'$. Recall again that $N(A_0)\cap N(Y'_2)\subseteq S_T$ by Claim~7. Thus we want to include in $B_T$ as many vertices as possible from the set of vertices that we have not yet assigned to $B_T$ or $S_T$, that is, from
the set $(Y_0' \setminus Y_2')\cup Z$.}

{
So our goal is to find a subset $X \subseteq Y_0' \setminus Y_2'$ 
and subsets $U' \subseteq U$ for each component $G[U]$ of $G[Z]$ to add to $B_T$.  We will first find each $U'$ and then $X$.   In each case, we will choose the largest possible subset such that the vertices added to $B_T$ so far do not induce an odd $T$-cycle.  We must ask whether this greedy approach will lead us to find the largest possible $B_T$ (given the initial choice of $A_0$ and $Y_2'$); that is, does the choice of some $U'$ not affect the later choice of another $U'$, or of $X$?}
 
{By Claim~13, we may consider the components of $G[Z]$ independently. Hence, we must just show when we choose a subset $U'$ that this does not affect our later choice of $X$ in a way that will prevent us claiming we have found a maximum size~$B_T$.
We will see that sometimes as we choose the largest possible $U'$, this forces us to put some vertices from 
$Y_0' \setminus Y_2'$ 
into $S_T$ (to avoid $G[B_T]$ containing an odd $T$-cycle); that is, we remove the possibility that such vertices can later be chosen to be in $X$.  However, when this happens, we will show that in choosing~$U'$, all the vertices of  
$Y_0' \setminus Y_2'$ 
forced into $S_T$ belong to the same component of $G[Y_0' \setminus Y_2']$. By Claim~12, any choice of $X$ will contain at most one vertex from this component. Hence, the \emph{cost} of our choice of $U'$ is that $X$ might contain one fewer vertex~$x$ than is otherwise possible. 
We can allow this as long as putting $x$ in $B_T$ prevents us
from putting more than one vertex of 
$U$ in $B_T$, 
and we show that this 
will always be the case.}

{Before we show how to find $U'$, let  us make one more comment.  When $U'$ is found, we must check that, for any choice of $X$,
it holds that $G[A_0\cup Y_2' \cup X\cup U']$ does not contain an odd $T$-cycle. First suppose that $U'\cap U_2\neq \emptyset$;
say $U'$ contains a vertex $u \in U_2$. By Claim~10, we have that $U'$ contains no other vertex of $U_2$. By Claim~11, no vertex of $U_1 \cap U'$ has neighbours in $B_T \setminus U'$.  Thus $u$ is a cutvertex of $G[B_T]$. Hence, we need only check that each of $G[A_0\cup Y_2' \cup X\cup \{u\}]$ and $G[U']$ contains no odd $T$-cycle.  As the former is bipartite, all that we need to confirm is that $G[U']$ does not contain an odd $T$-cycle, or, equivalently, since $U'$ is a clique, that $|U'| \leq 2$ or $U' \cap T = \emptyset$.}

{Now suppose that $U' \cap U_2 = \emptyset$. Then $U'$ is a subset of either $U_0$ or $U_1$, as at most one of the latter two sets is non-empty by Claim~9.
First suppose that $U'\subseteq U_0$.
Then there are no edges from $U'$ to the rest of $B_T$. 
So, again, we need only check that
$|U'| \leq 2$ or $U' \cap T = \emptyset$. 
Now suppose that $U'\subseteq U_1$.
Then the vertices of $U'$ might have neighbours in $B_T \setminus U'$.  These neighbours must lie in $Y_0'$, and, by Claim~9, they must form a clique so, by Claim~12, there can only be one such neighbour which we denote $y$.    Thus $y$ is a cutvertex so we need only check that the complete graph $G[U' \cup \{y\}]$ contains no odd $T$-cycle, or, equivalently, that either $U' \cup \{y\}$ contains only two vertices or no vertex of $T$.}

{We now show how to find $U'$ by
distinguishing three cases; note that we can check in polynomial time which case applies. We use again that $|U' \cap U_2| \le 1$ by Claim~10 and at most one of $U_0$ and $U_1$ is non-empty by Claim~9.  When we have found $U'$, we must make the checks outlined above.}

\medskip
\noindent
{{\bf Case~2bi.} $U_1 = \emptyset$.}\\ 
{We build $U'$ by first letting it contain all vertices of $U_0 \setminus T$ and, if $U_2 \setminus T$ is non-empty, one arbitrary vertex from the set $U_2 \setminus T$.  If $U'$ now has at least two vertices, then we are done since the only other vertices we could add are in $T$, but then $U'$ would induce an odd $T$-cycle.  If, however, $U'$ has fewer than two vertices,
then we add up to two vertices from $U \cap T$ subject to the constraints that $U'$ can only contain one vertex of $U_2$ and cannot contain more than two vertices (else, again, we obtain an odd $T$-cycle).  }
{In this way we have made $U'$ as large as possible. As $U' \cap U_1 = \emptyset$, we need only note that $U'$ contains no vertex of $T$, or has at most two vertices, to confirm that $G[B_T]$ will contain no odd $T$-cycle using a vertex of $U'$.}

\medskip
\noindent
{{\bf Case 2bii.} $U_1 \neq \emptyset$ and $N(U_1) \cap Y_0'$ is a subset of $Y_2'$.}\\
{By Claim 9, $N(U_1) \cap Y_0'$ is a clique. Hence, in this case, all the vertices of $N(U_1) \cap Y_0'$ belong to the same component of $G[Y_2']$.   As the components of $G[Y_2']$ are isolated vertices, we have that $N(U_1) \cap Y_0'=\{y\}$ for some $y \in Y_2'$.  As $Y_2' \subseteq B_T$, we cannot, by Claim 11, add vertices from both $U_1$ and $U_2$ to~$U'$. As we can only ever add one vertex from $U_2$, we will instead add only vertices from~$U_1$.  If $y \in T$ or $U_1 \subseteq T$, then we must let $U'$ contain only one vertex of $U_1$ else $U' \cup \{y\}$ will induce an odd $T$-cycle. Otherwise, we let $U'$ contain all vertices of $U_1 \setminus T$; there is at least one and so we cannot also add any vertex of $U_1 \cap T$ without $U' \cup \{y\}$ inducing an odd $T$-cycle.  We note that we have made $U'$ as large as possible and that $U' \cup \{y\}$ either contains no vertex of $T$ or contains only two vertices.}

\medskip
\noindent
{{\bf Case 2biii.} $U_1 \neq \emptyset$ and $N(U_1) \cap Y_0'$ is not a subset of $Y_2'$.}\\
{Again, by Claim 9, $N(U_1) \cap Y_0'$ is a clique. Hence, in this case, all the vertices of $N(U_1) \cap Y_0'$ belong to the same component of $G[Y_0'\setminus Y_2']$.
We build $U'$ by first letting it contain all vertices of $U_1 \setminus T$ and, if it is non-empty, one vertex from the set $U_2 \setminus T$.  If $U'$ now has at least two vertices, then we are done since we cannot add vertices of $T$.   If $U'$  has exactly two vertices and one is from $U_2$, then, if possible, we exchange that vertex for a vertex in $U_1 \cap T$.
  If $U'$ has fewer than two vertices, then we add vertices from $U \cap T$ noting again that $U'$ can only contain one vertex of $U_2$ and cannot now contain more than two vertices.  When we add these vertices, we give preference to vertices from $U_1$ and so only add a vertex from $U_2$ if $|U_1| = 1$.}
  
{Again we have made $U'$ as large as possible.  
If $U'$ contains vertices from each of $U_1$ and $U_2$, then note that, by Claim 11, $X$ will contain no vertices from $N(U_1) \cap Y_0'$.  Thus these vertices are added to $S_T$ and we will never add one of them to $X$.  Similarly, suppose that $U' \subseteq U_1$ and $|U'| \geq 2$: then $X$ can contain no vertex from 
$N(U_1) \cap Y_0'$ that is in $T$
and, moreover, can contain none at all from $N(U_1) \cap Y_0'$ if one of the vertices of $U'$ is in $T$.  Again we add vertices to $S_T$ as needed.
Note that, considering how $U'$ was built, we cannot avoid this without choosing a smaller $U'$: if it was possible to find a $U'$ of the same size without using a vertex of $U_2$ or in $T$ we would have done so. As at most one of the vertices we move to $S_T$ could ever have been added to $B_T$, there is no cost to this.}

{It only remains to note that if $U'$ contains a vertex of $U_2$, then it contains no vertex of $T$ unless it contains only two vertices, and that if $U' \cap U_2 = \emptyset$, then, for each neighbour $y \in Y_0'$ of $U'$ that has not been placed in $S_T$,
the set $\{y\} \cap U'$ contains no vertex of $T$ unless it has only two vertices.}

\medskip
\noindent
{We have completed the choice of each $U'$.  It now only remains to note that we can build $X$ by choosing one 
arbitrary
vertex from each component of $G[Y_0']$ excluding, of course, those vertices that we moved to $S_T$ in the process of finding each $U'$ 
(recall that we cannot add two vertices from a component of $G[Y_0']$ due to Claim~12).}

{This completes our proof. We conclude that we found a largest $T$-bipartite subgraph~{$G[B_T]$} in polynomial time. Hence, $S_T=V\setminus B_T$ is a smallest odd $T$-cycle transversal of $G$.} \qed
\end{proof}

We are now ready to prove our almost-complete classification.

\medskip
\noindent
{\bf Theorem~\ref{t-main2} (restated).}
{\it Let $H$ be a graph with $H\neq sP_1+P_4$ for all $s\geq 1$. 
Then {\sc Subset Odd Cycle Transversal} on $H$-free graphs is polynomial-time solvable if 
$H=P_4$ or $H\ssi sP_1+P_3$ for some $s\geq 1$, and is \NP-complete~otherwise.}

\begin{proof}
If $H$ has a cycle or claw, we use Theorem~\ref{t-known2}. The cases $H=P_4$ and $H=2P_2$ follow from
Theorems~\ref{soct-split} and~\ref{soct-p4}, respectively. The remaining case, where $H\ssi sP_1+P_3$, follows from Theorem~\ref{soct-sp1p3}.\qed
\end{proof}

\section{Conclusions}\label{s-con}

We gave almost-complete classifications of the complexity of {\sc Subset Feedback Vertex Set} and {\sc Subset Odd Cycle Transversal} for $H$-free graphs. 
The only open case in each classification is when $H=sP_1+P_4$ for some $s\geq 1$, which is also open for {\sc Feedback Vertex Set} and {\sc Odd Cycle Transversal} for $H$-free graphs.
Our proof techniques for $H=sP_1+P_3$ do not carry over and new structural insights are needed in order to solve the missing cases.

\begin{open}\label{o-1}
Determine the complexity of {\sc (Subset) Feedback Vertex Set} and {\sc (Subset) Odd Cycle Transversal} for $(sP_1+P_4)$-free graphs, when $s \ge 1$.
\end{open}

\noindent
One of the main obstacles to solving Open Problem~\ref{o-1} is the case where there is a solution $S$ such that $G-S$ is a forest that contains (many) arbitrarily large stars. In particular, 
\cref{tree-sp1p3} no longer holds.

\medskip
\noindent
The vertex-weighted version of {\sc Subset Feedback Vertex Set} has also been studied for $H$-free graphs. Papadopoulos and Tzimas~\cite{PT20} proved that {\sc Weighted Subset Feedback Vertex Set} is polynomial-time solvable for $4P_1$-free graphs but \NP-complete for $5P_1$-free graphs (in contrast to the unweighted version). 
Bergougnoux et al.~\cite{BPT19} recently proved that {\sc Weighted Subset Feedback Vertex Set} is polynomial-time solvable for every class of bounded mim-width and thus for $P_4$-free graphs. Combining these results with Theorem~\ref{t-main} 
and two recent results, polynomial-time solvable cases when $H\in \{P_1+P_3,3P_1+P_2\}$~\cite{BJP}, leaves three gaps (see also~\cite{BJP}).

\begin{open}\label{o-2}
Determine the complexity of {\sc Weighted Subset Feedback Vertex Set} for $H$-free graphs when $H\in \{2P_1+P_3,P_1+P_4, 2P_1+P_4\}$.
\end{open}

\noindent
For the weighted variant, a vertex in $T$ may have a large weight that prevents it from being deleted in any solution; in particular, \cref{bound}, 
which plays a crucial role in our proofs, 
no longer holds.

\medskip
\noindent
As shown in~\cite{BJP}, the \NP-completeness proof given by Papadopoulos and Tzimas for {\sc Weighted Subset Feedback Vertex Set} on $5P_1$-free graphs~\cite{PT20} can also be used to show that the weighted version of {\sc Subset Odd Cycle Transversal} is \NP-complete for $5P_1$-free graphs. 
In the same paper~\cite{BJP} it is proven that {\sc Weighted Subset Odd Cycle Transversal} is polynomial-time solvable for $H$-free graphs if $H\in\{P_1+P_3,3P_1+P_2,P_4\}$.
Combining these results with Theorem~\ref{t-main2} leads to the same three open cases as listed in Open Problem~\ref{o-2}.

\begin{open}\label{o-3}
Determine the complexity of {\sc Weighted Subset Odd Cycle Transversal} for  $H$-free graphs when $H\in \{2P_1+P_3,P_1+P_4, 2P_1+P_4\}$.
\end{open} 

\medskip
\noindent
As previously mentioned, Bergougnoux et al.~\cite{BPT19} gave an \XP\ algorithm for {\sc Weighted Subset Feedback Vertex Set} parameterized by mim-width.

\begin{open}\label{o-4}
Does there exist an \XP\ algorithm for {\sc Subset Odd Cycle Transversal} and {\sc Weighted Subset Odd Cycle Transversal} parameterized by mim-width?
\end{open}

\noindent
We also  introduced the {\sc Subset Vertex Cover} problem and showed that this problem is polynomial-time solvable on $(sP_1+P_4)$-free graphs for every $s\geq 0$.
Lokshtanov et al.~\cite{LVV14} proved that {\sc Vertex Cover} is polynomial-time solvable for $P_5$-free graphs. Grzesik et al.~\cite{GKPP19} extended this result to $P_6$-free graphs.

\begin{open}\label{o-5}
Determine the complexity of {\sc Subset Vertex Cover} for $P_5$-free graphs.
\end{open}

\begin{open}\label{o-6}
Determine whether there exists an integer~$r\geq 5$ such that {\sc Subset Vertex Cover} is \NP-complete for $P_r$-free graphs.
\end{open}

\noindent
By Poljak's construction~\cite{Po74}, {\sc Vertex Cover} is \NP-complete for $H$-free graphs if~$H$ has a cycle.
However, {\sc Vertex Cover} becomes polynomial-time solvable on $K_{1,3}$-free graphs~\cite{Mi80,Sh80}. 
We did not research the complexity of {\sc Subset Vertex Cover} on $K_{1,3}$-free graphs and also leave this as an open problem for future work.

\begin{open}\label{o-7}
Determine the complexity of {\sc Subset Vertex Cover} for $K_{1,3}$-free graphs.
\end{open}

\noindent
Finally, several related transversal problems have been studied but not yet for $H$-free graphs. For example, the parameterized complexity of
{\sc Even Cycle Transversal} and {\sc Subset Even Cycle Transversal} has been addressed in~\cite{MRRS12} and~\cite{KKK12}, respectively.
Moreover, several variants for transversal problems have been studied for $H$-free graphs, 
but not the subset version:
for example,  
{\sc Connected Vertex Cover}, {\sc Connected Feedback Vertex Set} and {\sc Connected Odd Cycle Transversal}, and also for {\sc Independent Feedback Vertex Set} and {\sc Independent Odd Cycle Transversal}; see \cite{BDFJP19,CHJMP18,DJPPZ18,JPP20} for a number of recent results.
It would be interesting to solve the subset versions of these transversal problems for $H$-free graphs and to determine the connections amongst all these problems in a more general framework.

\medskip
\noindent
{\it Acknowledgments.} We thank three anonymous reviewers for helpful comments  that improved the presentation of several proofs.


\begin{thebibliography}{10}

\bibitem{ACPRS20}
T.~Abrishami, M.~Chudnovsky, M.~Pilipczuk, P.~Rz{\k{a}}\.{z}ewski, and
  P.~Seymour.
\newblock Induced subgraphs of bounded treewidth and the container method.
\newblock {\em Proc. SODA 2021}, 1948--1964, 2021.

\bibitem{BBBK20}
B.~Bergougnoux, {\'E}.~Bonnet, N.~Brettell, and O.~Kwon.
\newblock Close relatives of feedback vertex set without single-exponential
  algorithms parameterized by treewidth.
\newblock {\em Proc. IPEC 2020, LIPIcs}, 180:3:1--3:17, 2020.

\bibitem{BPT19}
B.~Bergougnoux, C.~Papadopoulos, and J.~A. Telle.
\newblock Node multiway cut and subset feedback vertex set on graphs of bounded
  mim-width.
\newblock {\em Proc. WG 2020, LNCS}, 12301:388--400, 2020.

\bibitem{BM93}
H.~L. Bodlaender and R.~H. M{\"{o}}hring.
\newblock The pathwidth and treewidth of cographs.
\newblock {\em {SIAM} Journal on Discrete Mathematics}, 6:181--188, 1993.

\bibitem{BDFJP19}
M.~Bonamy, K.~K. Dabrowski, C.~Feghali, M.~Johnson, and D.~Paulusma.
\newblock Independent feedback vertex set for {$P_5$}-free graphs.
\newblock {\em Algorithmica}, 81:1342--1369, 2019.

\bibitem{BK85}
A.~Brandst{\"a}dt and D.~Kratsch.
\newblock On the restriction of some {NP-complete} graph problems to
  permutation graphs.
\newblock {\em Proc. FCT 1985, LNCS}, 199:53--62, 1985.

\bibitem{BLS99}
A.~Brandst{\"a}dt, V.~B. Le, and J.~P. Spinrad.
\newblock {\em Graph Classes: A Survey}, volume~3 of {\em SIAM Monographs on
  Discrete Mathematics and Applications}.
\newblock SIAM, 1999.

\bibitem{BHMPP}
N.~Brettell, J.~Horsfield, A.~Munaro, G.~Paesani, and D.~Paulusma.
\newblock Bounding the mim-width of hereditary graph classes.
\newblock {\em Proc. IPEC 2020, LIPIcs}, 180:6:1--6:18, 2020.

\bibitem{BJPP20}
N.~Brettell, M.~Johnson, G.~Paesani, and D.~Paulusma.
\newblock Computing subset transversals in ${H}$-free graphs.
\newblock {\em Proc. WG 2020, LNCS}, 12301:187--199, 2020.

\bibitem{BJP}
N.~Brettell, M.~Johnson, and D.~Paulusma.
\newblock Computing weighted subset transversals in ${H}$-free graphs.
\newblock {\em Proc. WADS 2021, LNCS}, 12808:229--242, 2021.

\bibitem{CHJMP18}
N.~Chiarelli, T.~R. Hartinger, M.~Johnson, M.~Milani\v{c}, and D.~Paulusma.
\newblock Minimum connected transversals in graphs: New hardness results and
  tractable cases using the price of connectivity.
\newblock {\em Theoretical Computer Science}, 705:75--83, 2018.

\bibitem{CFLMRS17}
R.~Chitnis, F.~V. Fomin, D.~Lokshtanov, P.~Misra, M.~S. Ramanujan, and
  S.~Saurabh.
\newblock Faster exact algorithms for some terminal set problems.
\newblock {\em Journal of Computer and System Sciences}, 88:195--207, 2017.

\bibitem{CLS81}
D.~G. Corneil, H.~Lerchs, and L.~S. Burlingham.
\newblock Complement reducible graphs.
\newblock {\em Discrete Applied Mathematics}, 3:163--174, 1981.

\bibitem{CPS85}
D.~G. Corneil, Y.~Perl, and L.~K. Stewart.
\newblock A linear recognition algorithm for cographs.
\newblock {\em {SIAM} Journal on Computing}, 14:926--934, 1985.

\bibitem{CPPW13}
M.~Cygan, M.~Pilipczuk, M.~Pilipczuk, and J.~O. Wojtaszczyk.
\newblock Subset feedback vertex set is fixed-parameter tractable.
\newblock {\em {SIAM} Journal on Discrete Mathematics}, 27:290--309, 2013.

\bibitem{DFJPPP19}
K.~K. Dabrowski, C.~Feghali, M.~Johnson, G.~Paesani, D.~Paulusma, and
  P.~Rz{\k{a}}\.{z}ewski.
\newblock On cycle transversals and their connected variants in the absence of
  a small linear forest.
\newblock {\em Algorithmica}, 82:2841--2866, 2020.

\bibitem{DJPPZ18}
K.~K. Dabrowski, M.~Johnson, G.~Paesani, D.~Paulusma, and V.~Zamaraev.
\newblock On the price of independence for vertex cover, feedback vertex set
  and odd cycle transversal.
\newblock {\em Proc. {M}{F}{C}{S} 2018, LIPIcs}, 117:63:1--63:15, 2018.

\bibitem{FH77}
S.~F\"oldes and P.~L. Hammer.
\newblock Split graphs.
\newblock {\em Congressus Numerantium}, 19:311--315, 1977.

\bibitem{FHKPV14}
F.~V. Fomin, P.~Heggernes, D.~Kratsch, C.~Papadopoulos, and Y.~Villanger.
\newblock Enumerating minimal subset feedback vertex sets.
\newblock {\em Algorithmica}, 69:216--231, 2014.

\bibitem{GHKS14}
P.~A. Golovach, P.~Heggernes, D.~Kratsch, and R.~Saei.
\newblock Subset feedback vertex sets in chordal graphs.
\newblock {\em Journal of Discrete Algorithms}, 26:7--15, 2014.

\bibitem{GKPP19}
A.~Grzesik, T.~Klimo\v{s}ov\'a, M.~Pilipczuk, and M.~Pilipczuk.
\newblock Polynomial-time algorithm for maximum weight independent set on
  {$P_6$-free} graphs.
\newblock {\em Proc. SODA 2019}, 1257--1271, 2019.

\bibitem{HK18}
E.~C. Hols and S.~Kratsch.
\newblock A randomized polynomial kernel for subset feedback vertex set.
\newblock {\em Theory of Computing Systems}, 62:63--92, 2018.

\bibitem{HT73}
J.~E. Hopcroft and R.~E. Tarjan.
\newblock Algorithm 447: Efficient algorithms for graph manipulation.
\newblock {\em Communications of the {ACM}}, 16:372--378, 1973.

\bibitem{IWY16}
Y.~Iwata, M.~Wahlstr{\"{o}}m, and Y.~Yoshida.
\newblock Half-integrality, lp-branching, and {FPT} algorithms.
\newblock {\em {SIAM} Journal on Computing}, 45:1377--1411, 2016.

\bibitem{JKT20}
L.~Jaffke, O.~Kwon, and J.~A. Telle.
\newblock Mim-width {II.} the feedback vertex set problem.
\newblock {\em Algorithmica}, 82:118--145, 2020.

\bibitem{JPP20}
M.~Johnson, G.~Paesani, and D.~Paulusma.
\newblock Connected vertex cover for $(s{P}_1+{P}_5)$-free graphs.
\newblock {\em Algorithmica}, 82:20--40, 2020.

\bibitem{KKK12}
N.~Kakimura, K.~Kawarabayashi, and Y.~Kobayashi.
\newblock Erd{\H{o}}s-{P\'{o}}sa property and its algorithmic applications:
  parity constraints, subset feedback set, and subset packing.
\newblock {\em Proc. {SODA} 2012}, 1726--1736, 2012.

\bibitem{KK12}
K.~Kawarabayashi and Y.~Kobayashi.
\newblock Fixed-parameter tractability for the subset feedback set problem and
  the {S}-cycle packing problem.
\newblock {\em Journal of Combinatorial Theory, Series B}, 102:1020--1034,
  2012.

\bibitem{KW12}
S.~Kratsch and M.~Wahlstr{\"{o}}m.
\newblock Representative sets and irrelevant vertices: New tools for
  kernelization.
\newblock {\em Proc. {FOCS} 2012}, 450--459, 2012.

\bibitem{LMRS17}
D.~Lokshtanov, P.~Misra, M.~S. Ramanujan, and S.~Saurabh.
\newblock Hitting selected (odd) cycles.
\newblock {\em {SIAM} Journal on Discrete Mathematics}, 31:1581--1615, 2017.

\bibitem{LVV14}
D.~Lokshtanov, M.~Vatshelle, and Y.~Villanger.
\newblock Independent set in ${P}_5$-free graphs in polynomial time.
\newblock {\em Proc. {SODA} 2014},  570--581, 2014.

\bibitem{Mi80}
G.~J. Minty.
\newblock On maximal independent sets of vertices in claw-free graphs.
\newblock {\em Journal of Combinatorial Theory, Series B}, 28:284--304,
  1980.

\bibitem{MRRS12}
P.~Misra, V.~Raman, M.~S. Ramanujan, and S.~Saurabh.
\newblock Parameterized algorithms for even cycle transversal.
\newblock {\em Proc. {WG} 2012}, 7551:172--183, 2012.

\bibitem{Mu17b}
A.~Munaro.
\newblock On line graphs of subcubic triangle-free graphs.
\newblock {\em Discrete Mathematics}, 340:1210--1226, 2017.

\bibitem{PT19}
C.~Papadopoulos and S.~Tzimas.
\newblock Polynomial-time algorithms for the subset feedback vertex set problem
  on interval graphs and permutation graphs.
\newblock {\em Discrete Applied Mathematics}, 258:204--221, 2019.

\bibitem{PT20}
C.~Papadopoulos and S.~Tzimas.
\newblock Subset feedback vertex set on graphs of bounded independent set size.
\newblock {\em Theoretical Computer Science}, 814:177--188, 2020.

\bibitem{Po74}
S.~Poljak.
\newblock A note on stable sets and colorings of graphs.
\newblock {\em Commentationes Mathematicae Universitatis Carolinae},
  15:307--309, 1974.

\bibitem{Sh80}
N.~Sbihi.
\newblock Algorithme de recherche d'un stable de cardinalit\'e maximum dans un
  graphe sans \'etoile.
\newblock {\em Discrete Mathematics}, 29:53--76, 1980.

\end{thebibliography}
\end{document}